\documentclass[11pt]{article}%
\usepackage{amsfonts}
\usepackage{amsmath}
\usepackage{amssymb}
\usepackage{graphicx}%
\providecommand{\U}[1]{\protect\rule{.1in}{.1in}}
\newtheorem{theorem}{Theorem}
\newtheorem{acknowledgement}[theorem]{Acknowledgement}

\newtheorem{definition}[theorem]{Definition}

\newtheorem{lemma}[theorem]{Lemma}

\newtheorem{remark}[theorem]{Remark}

\newenvironment{proof}[1][Proof]{\noindent\textbf{#1.} }{\ \rule{0.5em}{0.5em}}
\begin{document}

\title{Quantum Polar Duality and the Symplectic Camel: a New Geometric Approach to Quantization}
\author{Maurice A. de Gosson\thanks{maurice.de.gosson@univie.ac.at}\\University of Vienna\\Faculty of Mathematics (NuHAG)}
\maketitle
\tableofcontents

\begin{abstract}
We define and study the notion of quantum polarity, which is a kind of
geometric Fourier transform between sets of positions and sets of momenta.
Extending previous work of ours, we show that the orthogonal projections of
the covariance ellipsoid of a quantum state on the configuration and momentum
spaces form what we call a dual quantum pair. We thereafter show that quantum
polarity allows solving the Pauli reconstruction problem for Gaussian
wavefunctions. The notion of quantum polarity exhibits a strong interplay
between the uncertainty principle and symplectic and convex geometry and our
approach could therefore pave the way for a geometric and topological version
of quantum indeterminacy. We relate our results to the Blaschke-Santal\'{o}
inequality and to the Mahler conjecture. We also discuss the Hardy uncertainty
principle and the less-known Donoho--Stark principle from the point of view of
quantum polarity.

\end{abstract}


\section{Introduction}

The notion of duality is omnipresent in science and philosophy, and in human
thinking \cite{Maurin}. Duality in science is usually implemented using a
transformation which serves as a dictionary for translating between two
different representations of an object. In quantum mechanics this role is
played by the Fourier transform which allows one to switch from the position
representation to the momentum representation. In this article we introduce a
new kind of duality in quantum mechanics, having its roots in convex geometry.
While the Fourier transform turns a function in $x$-space into a function in
$p$-space our duality turns a set of positions into a set of momenta: it is
thus a kind of proto-Fourier transform operating between sets, and not
functions. The definition of this duality is actually very simple and it is
therefore somewhat surprising that it hasn't been noticed or used earlier in
the literature. It goes as follows: let $X$ be a convex body in configuration
space $\mathbb{R}_{x}^{n}$; we assume that $X$ contains the origin. This set
may be, for instance, the convex closure of a cloud of position measurements
performed on some physical system located near the origin. To $X$ we associate
its \textit{polar dual} $X^{\hbar}$. It is, by definition, the set of all
points $p=(p_{1},...,p_{n})$ in momentum space $\mathbb{R}_{p}^{n}$ such that
we have
\[
p_{1}x_{1}+\cdot\cdot\cdot+p_{n}x_{n}\leq\hbar
\]
for all values $x=(x_{1},...,x_{n})$ in $X$. It turns out that the
correspondence $X\longleftrightarrow X^{\hbar}$ is in a sense a geometric
variant of the correspondence $\psi\longleftrightarrow\widehat{\psi}$ between
a wavefunction $\psi$ and its Fourier transform and thus contains the
uncertainty principle in disguise. Somewhat oversimplifying, we could say that:

\begin{quotation}
\emph{A quantum system localized in the position representation in a set} $X$
\emph{cannot be localized in the momentum representation in a set smaller than
its polar dual} $X^{\hbar}$.
\end{quotation}

The following simple example illustrates this interpretation. Consider a pure
quantum state $|\psi\rangle$ on the $x$ axis; we assume for simplicity that
the state is centered at $\langle x\rangle=\langle p\rangle=0$. That state has
a covariance matrix
\begin{equation}
\Sigma=%
\begin{pmatrix}
\sigma_{xx} & \sigma_{xp}\\
\sigma_{px} & \sigma_{pp}%
\end{pmatrix}
\text{ \ },\text{ \ }\sigma_{xp}=\sigma_{px}%
\end{equation}
where $\sigma_{xx}=\langle\widehat{x}^{2}\rangle$, $\sigma_{pp}=\langle
\widehat{p}^{2}\rangle$, and $\sigma_{xp}=\frac{1}{2}\langle\widehat{x}%
\widehat{p}+\widehat{p}\widehat{x}\rangle$. \ The determinant of $\Sigma$ is
$D=\sigma_{xx}\sigma_{pp}-\sigma_{xp}^{2}$ and in view of the uncertainty
principle in its strong form (the Robertson--Schr\"{o}dinger inequality) we
must have $D\geq\tfrac{1}{4}\hbar^{2}$. We associate with $\Sigma$ the
covariance ellipse $\Omega$: it is the set of all points $z=(x,p)$ in the
phase plane such that $\frac{1}{2}\Sigma^{-1}z\cdot z\leq1$; in the
coordinates $x,p$
\begin{equation}
\Omega:~\dfrac{\sigma_{pp}}{2D}x^{2}-\frac{\sigma_{xp}}{D}px+\dfrac
{\sigma_{xx}}{2D}p^{2}\leq1~. \label{2D}%
\end{equation}
The orthogonal projections $\Omega_{X}$ and $\Omega_{P}$ of $\Omega$ on the
$x$ and $p$ axes are the intervals
\begin{equation}
\Omega_{X}=[-\sqrt{2\sigma_{xx}},\sqrt{2\sigma_{xx}}]\text{\ \textit{\ , \ }
}\Omega_{P}=[-\sqrt{2\sigma_{pp}},\sqrt{2\sigma_{pp}}\dot{]}~.
\label{intervals}%
\end{equation}
Let $\Omega_{X}^{\hbar}$ be the polar dual of $\Omega_{X}$: it is the set of
all numbers $p$ such that%
\[
px\leq\hbar\text{ \ \textit{for} }-\sqrt{2\sigma_{xx}}\leq x\leq\sqrt
{2\sigma_{xx}}%
\]
and is thus the interval
\[
\Omega_{X}^{\hbar}=[-\hbar/\sqrt{2\sigma_{xx}},\hbar/\sqrt{2\sigma_{xx}}]~.
\]
We make the following crucial observation: since $\sigma_{xx}\sigma_{pp}%
\geq\frac{1}{2}\hbar$ we have the inclusion
\begin{equation}
\Omega_{X}^{\hbar}\subset\Omega_{P} \label{xompom}%
\end{equation}
and this inclusion reduces to the equality $\Omega_{X}^{\hbar}=\Omega_{P}$ if
and only if the Heisenberg inequality is saturated (\textit{i.e.} $\sigma
_{xx}\sigma_{pp}=\frac{1}{4}\hbar^{2}$); this corresponds to the case where
the state $|\psi\rangle$ is the minimum uncertainty Gaussian
\[
\psi_{0}(x)=\left(  \tfrac{1}{2\pi\sigma_{xx}}\right)  ^{1/4}e^{-\frac{x^{2}%
}{4\sigma_{xx}}}~.
\]
This example suggests that the uncertainty principle (UP) can be expressed
using polar duality, which is a tool from convex geometry. In fact, we will
see that it allows for a more general expression of the UP, of which the
traditional one, using variances and covariances is a particular case.

In the present work we will extend this discussion to states with arbitrary
numbers of freedoms; the approach we outline is both simple and subtle and is
closely related to open problems in geometry (the Mahler conjecture, Section
\ref{secmahl}). We will see that the notion of quantum polarity is not only
important from a foundational point of view, but also very fruitful for
solving \textquotedblleft practical\textquotedblright\ problems. For instance
we will show that it plays an essential role for the understanding and
resolution of Pauli's reconstruction problem \cite{Pauli} (Theorems \ref{Thm1}
and \ref{Thm2}). Historically, the problem goes back to the famous question
Pauli asked in \cite{Pauli}, whether the probability densities $|\psi(x)|^{2}$
and $|\widehat{\psi}(p)|^{2}$ uniquely determine the wavefunction $\psi(x)$.

On a more conceptual level, it turns out that the properties of quantum polar
duality can be reformulated in terms of a notion from symplectic topology, the
\textquotedblleft principle of the symplectic camel\textquotedblright. In
\cite{go09} we already suggested that this deep and surprising principle might
well be the \textquotedblleft tip of an iceberg\textquotedblright. Here we go
a few steps further; our analysis in \cite{go09} was based on the usual
formulation of the uncertainty principle in terms of (co-)variances of quantum
observables, which has a long story following the work of Heisenberg,
Schr\"{o}dinger, Weyl, Kennard, Robertson and many others. However, as pointed
our by several authors, standard deviations only give adequate measurements of
the spread for states that are Gaussian, or close to Gaussian states
(Hilgevoord and Uffink \cite{hi02,hiuf85bis}, Sharma \textit{et al}.
\cite{Sharma}; also see Butterfield's interesting analysis \cite{Butter}). It
seems to us that the more geometric approach outlined in the present paper
helps to avoid this pitfall. Even if some of the consequences of polar duality
can be stated in terms of covariance matrices and standard uncertainty
principles, these appear as secondary objects: the use of quantum dual pairs
liberates the UP from the traditional use of tools from classical statistics
and probability theory, such as variances and covariances.

In previous work \cite{blob} the expression of the UP in terms of covariance
matrices and ellipsoids led us to define the notion of \textquotedblleft
quantum blob\textquotedblright, the smallest unit of phase space allowed by
the UP in its traditional Robertson--Schr\"{o}dinger form. The
\textquotedblleft philosophy\textquotedblright\ behind the introduction of
quantum polar dual pairs is the following: instead of talking about
measurements and their statistical properties (which are always related to
some underlying quasi-probability distribution, we proceed purely
geometrically by associating to every convex body $X$ in position space its
quantum polar dual $X^{\hbar}$; the product $X\times X^{\hbar}$ then forms a
kind of\ phase space \textquotedblleft quantum cell\textquotedblright, always
containing\ a quantum blob, but the definition of $X\times X^{\hbar}$, as
opposed to that of quantum blobs, is independent of any particular given
state. In a sense, this new kind of \textquotedblleft
coarse-graining\textquotedblright\ may be more physical since the primary
object, $X$, is a subset of the physical space $\mathbb{R}_{x}^{n}$ which is
\textquotedblleft Fourier transformed\textquotedblright\ by polar duality into
a subset of momentum space $\mathbb{R}_{p}^{n}$ as in traditional quantum
mechanics, where one associates to a wavefunction its Fourier transform. But
there is no wavefunction here!

\paragraph*{Notation and terminology}

We denote by $\mathbb{R}_{z}^{2n}\equiv\mathbb{R}_{x}^{n}\times\mathbb{R}%
_{p}^{n}$ the phase system of a system with $n$ degrees of freedom; it comes
equipped with the standard symplectic form $\omega(z,z^{\prime})=Jz\cdot
z^{\prime}=(z^{\prime})^{T}Jz$ where
\[
J=%
\begin{pmatrix}
0_{n\times n} & I_{n\times n}\\
-I_{n\times n} & 0_{n\times n}%
\end{pmatrix}
\]
is the standard symplectic matrix. The symplectic group associated with
$\omega$ is denoted by $\operatorname*{Sp}(n)$; it consists of all linear
automorphisms $S$ of phase space such that $\omega(Sz,Sz^{\prime}%
)=\omega(z,z^{\prime})$ for all $z,z^{\prime}$ in $\mathbb{R}_{z}^{2n}$;
equivalently $S^{T}JS=SJS^{T}=J$. The metaplectic group $\operatorname*{Mp}%
(n)$ is a group of unitary operators on $L^{2}(\mathbb{R}^{n})$ which is a
double covering of $\operatorname*{Sp}(n)$: to every $S\in\operatorname*{Sp}%
(n)$ correspond two operators $\pm\widehat{S}\in\operatorname*{Mp}(n)$. We
denote by $\operatorname*{Symp}(n)$ the group of all canonical transformations
(= symplectomorphisms) of $(\mathbb{R}_{z}^{2n},\omega)$: $f\in
\operatorname*{Symp}(n)$ if and only if $f$ is a diffeomorphism of
$\mathbb{R}_{z}^{2n}$ and $f^{\ast}\omega=\omega$; equivalently $f$ is
bijective, infinitely differentiable and with infinitely differentiable
inverse, and the Jacobian matrix $Df(z)$ is symplectic for every $z$.

We will use the L\"{o}wner partial ordering of matrices \cite{loewnerorder}%
:$A\geq B$ (\textit{resp}. $A>B$) means that $A-B$ is positive semidefinite
(\textit{resp}. positive definite). When writing $A>0$ it is always understood
that $A=A^{T}$ ($A^{T}$ the transpose of $A$).

The $n$-dimensional Fourier transform $\widehat{\psi}=F\psi$ of $\psi\in
L^{2}(\mathbb{R}^{n})$ is defined for $\psi\in L^{1}(\mathbb{R}^{n})\cap
L^{2}(\mathbb{R}^{n})$ by
\begin{equation}
\widehat{\psi}(p)=\left(  \tfrac{1}{2\pi\hbar}\right)  ^{n/2}\int e^{-\frac
{i}{\hbar}px}\psi(x)d^{n}x~. \label{Fourier2}%
\end{equation}

\section{Background Material\label{secone}}

We begin by recalling the main properties of density matrices; for a detailed
rigorous review see \cite{QUANTA}. We thereafter introduce the basic notions
from harmonic analysis that we will need.

\subsection{Density matrices and their covariance ellipsoids}

We recall some material about the density matrix formalism following our
presentation in \cite{QUANTA}.

\subsubsection{Density matrices and their Wigner distributions}

Mixed quantum states will be as usual identified with their density matrices
which are convex sums of projection operators on rays $\mathbb{C}\psi_{j}$
\[
\widehat{\rho}=\sum_{j}\lambda_{j}|\psi_{j}\rangle\langle\psi_{j}|~.
\]
A quantum state $\widehat{\rho}$ on $L^{2}(\mathbb{R}_{x}^{n})$ is a positive
semidefinite $\widehat{\rho}\geq0$ (and hence self-adjoint) operator on
$L^{2}(\mathbb{R}_{x}^{n})$\ with trace $\operatorname*{Tr}(\widehat{\rho}%
)=1$. Such an operator is always compact and hence bounded. By definition the
Wigner distribution of the state $\widehat{\rho}$ is the function
$W_{\widehat{\rho}}\in L^{2}(\mathbb{R}_{z}^{2n})$ defined by
\begin{equation}
W_{\widehat{\rho}}=\sum_{j}\lambda_{j}W\psi_{j} \label{wig}%
\end{equation}
where $W\psi_{j}$ is the usual Wigner transform of $\psi_{j}$: for $\psi\in
L^{2}(\mathbb{R}^{n})$
\begin{equation}
W\psi(x,p)=\left(  \tfrac{1}{2\pi\hbar}\right)  ^{n}\int e^{-\frac{i}{\hbar
}py}\psi(x+\tfrac{1}{2}y)\psi^{\ast}(x-\tfrac{1}{2}y)d^{n}y~. \label{wigtra}%
\end{equation}
The Wigner distribution of $\widehat{\rho}$ is conventionally written in
bra-ket notation%
\begin{equation}
W_{\widehat{\rho}}(z)=\left(  \tfrac{1}{2\pi\hbar}\right)  ^{n}\int
e^{-\frac{i}{\hbar}py}\left\langle x+\tfrac{1}{2}y\right\vert \widehat{\rho
}\left\vert x-\tfrac{1}{2}y\right\rangle d^{n}y \label{robraket}%
\end{equation}
but we will not use this notation.

\subsubsection{The covariance matrix and ellipsoid}

Assuming that $W\psi_{j}\in L^{1}(\mathbb{R}^{n})\cap L^{2}(\mathbb{R}^{n})$
is $L^{2}$ normalized for each $j$ the Wigner distribution $W_{\widehat{\rho}%
}(z)$ plays the role of a quasi probability distribution on phase space; this
is illustrated by the marginal properties
\begin{align}
\int W_{\widehat{\rho}}(z)d^{n}p  &  =\sum_{j}\lambda_{j}|\psi_{j}%
(x)|^{2}\label{marg1}\\
\int W_{\widehat{\rho}}(z)d^{n}x  &  =\sum_{j}\lambda_{j}|\widehat{\psi_{j}%
}(p)|^{2}~. \label{marg2}%
\end{align}
Assuming in addition that the $W\psi_{j}$ decrease sufficiently fast at
infinity to ensure the existence of first and second moments, one defines the
covariance matrix of $\widehat{\rho}$ by
\begin{equation}
\Sigma=%
\begin{pmatrix}
\Sigma_{XX} & \Sigma_{XP}\\
\Sigma_{PX} & \Sigma_{PP}%
\end{pmatrix}
\text{ \ },\text{ \ }\Sigma_{PX}=\Sigma_{XP}^{T} \label{defcovma}%
\end{equation}
with $\Sigma_{XX}=(\sigma_{x_{j}x_{k}})_{1\leq j,k\leq n}$, $\Sigma
_{PP}=(\sigma_{p_{j}p_{k}})_{1\leq j,k\leq n}$, and $\Sigma_{XP}%
=(\sigma_{x_{j}p_{k}})_{1\leq j,k\leq n}$. Assuming for notational simplicity
that the first moments vanish%
\begin{equation}
\int x_{j}W_{\widehat{\rho}}(z)d^{2n}z=\int p_{j}W_{\widehat{\rho}}%
(z)d^{2n}z=0 \label{average}%
\end{equation}
the covariances $\sigma_{x_{j}x_{k}}$ are defined by the integrals
\begin{equation}
\sigma_{x_{j}x_{k}}=\int x_{j}x_{k}W_{\widehat{\rho}}(z)d^{2n}z \label{cojk}%
\end{equation}
and similar formulas for $\sigma_{x_{j},p_{k}}$ and $\sigma_{p_{j},p_{k}}$. In
more compact form,
\begin{equation}
\Sigma=\int zz^{T}W_{\widehat{\rho}}(z)d^{2n}z \label{coz}%
\end{equation}
where $z$, $x$ and $p$ are viewed as column vectors. A crucial fact
\cite{dutta,Birk,Narcow1,sisumu} is that the covariance matrix $\Sigma$
satisfies the \textquotedblleft quantum condition\textquotedblright%
\begin{equation}
\Sigma+\frac{i\hbar}{2}J\geq0~. \label{Quantum}%
\end{equation}
This condition implies in particular that $\Sigma>0$ \cite{Narcow1} (hence
$\Sigma$ is invertible). Condition (\ref{Quantum}) is necessary (but not
sufficient except in the Gaussian case \cite{golubis}) for the positivity
condition $\widehat{\rho}\geq0$ to hold \cite{go09,golu09}, and implies the
Robertson--Schr\"{o}dinger uncertainty principle (RSUP)%
\begin{equation}
\sigma_{x_{j}x_{j}}\sigma_{p_{j}p_{j}}\geq\sigma_{x_{j}p_{j}}^{2}+\tfrac{1}%
{4}\hbar^{2} \label{RS}%
\end{equation}
for $1\leq j\leq n$. To see this it suffices to use Sylvester's criterion for
the leading principal minors of a positive matrix, which implies that we must
have
\[%
\begin{vmatrix}
\sigma_{x_{j}x_{j}} & \sigma_{x_{j}p_{j}}+\frac{i\hbar}{2}\\
\sigma_{x_{j}p_{j}}-\frac{i\hbar}{2} & \sigma_{p_{j}p_{j}}%
\end{vmatrix}
>0
\]
which is equivalent to (\ref{RS}). Consider now the covariance ellipsoid of
$\widehat{\rho}$; it is the phase space ellipsoid%
\begin{equation}
\Omega=\{z:\frac{1}{2}\Sigma^{-1}z\cdot z\leq1\} \label{covell0}%
\end{equation}
where we are using the notation $\Sigma^{-1}z\cdot z=z^{T}\Sigma^{-1}z$. We
have proven in \cite{go09} that the conditions (\ref{Quantum}), (\ref{RS}) are
equivalent to the following statement:
\begin{equation}
\text{\textit{There exists} }S\in\operatorname*{Sp}(n)\text{ \textit{such
that} }S(\mathcal{B}^{2n}(\sqrt{\hbar}))\subset\Omega\label{camel2}%
\end{equation}
where $\mathcal{B}^{2n}(\sqrt{\hbar})$ is the phase space ball with radius
$\sqrt{\hbar}$ centered at the origin; this condition can in turn be rephrased
in terms of the topological notion of symplectic capacity (see Section
\ref{secsymp}). In \cite{blob} (also see \cite{Bull}) we have called the
minimum uncertainty ellipsoids $S(\mathcal{B}^{2n}(\sqrt{\hbar}))$
\textquotedblleft quantum blobs\textquotedblright\ hence the quantum condition
(\ref{Quantum}) amounts to saying that
\begin{equation}
\text{\textit{The covariance ellipsoid} }\Omega\text{ \textit{contains a
quantum blob.}} \label{qb}%
\end{equation}

\subsection{Symplectic and metaplectic covariance properties}

We are using Weyl's quantization scheme\ (Weyl correspondence) in this paper.
One of its hallmarks is its symplectic/metaplectic covariance properties.

\subsubsection{Symplectic covariance}

Density matrices and their Wigner distribution enjoy a covariance property
with respect to (linear) symplectic transformations. The idea is that if we
make a symplectic change of coordinates, the effect is that the corresponding
density operator will be changed by conjugation with any one of the two
associated metaplectic operators. More precisely, let us write
$W_{\widehat{\rho}}\leftrightarrows\widehat{\rho}$ the one-to-one
correspondence between Wigner distributions and the corresponding density
matrices. Then \cite{Birk,Birkbis,QUANTA,Littlejohn}, for every $S\in
\operatorname*{Sp}(n)$%
\[
W_{\widehat{\rho}}\circ S^{-1}\leftrightarrows\widehat{S}\widehat{\rho
}\widehat{S}^{-1}%
\]
where $\pm\widehat{S}$ $\in\operatorname*{Mp}(n)$ corresponds to $S$. In
particular, when $\widehat{\rho}$ describes a pure state $|\psi\rangle$ this
becomes%
\[
W\psi(S^{-1}z)=W(\widehat{S}\psi)(z)~.
\]
These formulas are actually particular cases of the general symplectic
covariance property of Weyl calculus, which plays an essential role in the
study of the symmetry properties of quantization.

The symplectic covariance property allows one to describe the action of
symplectic transformations on the covariance ellipsoid $\Omega$ in terms of
the state $\widehat{\rho}$ and its Wigner distribution $W_{\widehat{\rho}}$.
The following table summarizes these properties%
\begin{equation}%
\begin{tabular}
[c]{|l|l|l|l|}\hline
$\Omega$ & $\Sigma$ & $W_{\widehat{\rho}}$ & $\widehat{\rho}$\\\hline
$S\Omega$ & $S\Sigma S^{T}$ & $W_{\widehat{\rho}}\circ S^{-1}$ &
$\widehat{S}\widehat{\rho}\widehat{S}^{-1}$\\\hline
\end{tabular}
\ \ ~. \label{Table1}%
\end{equation}

\subsubsection{The generators of $\operatorname*{Sp}(n)$ and
$\operatorname*{Mp}(n)$\label{secgen}}

For practical purposes, let us describe a simple class of generators of
$\operatorname*{Mp}(n)$. Defining, for symmetric $P$ and invertible $L$,
\begin{equation}
V_{-P}=%
\begin{pmatrix}
I_{n\times n} & 0_{n\times n}\\
P & I_{n\times n}%
\end{pmatrix}
\text{ \ },\text{ \ }M_{L}=%
\begin{pmatrix}
L^{-1} & 0_{n\times n}\\
0_{n\times n} & L^{T}%
\end{pmatrix}
\label{vpml}%
\end{equation}
the symplectic group $\operatorname*{Sp}(n)$ is generated by the set of all
matrices $V_{-P}$ and $M_{L}$ together with the standard symplectic matrix
$J$. To these generators of $\operatorname*{Sp}(n)$ correspond the generators
$\pm\widehat{V}_{-P}$, $\pm\widehat{M}_{L,m}$, and $\pm\widehat{J}$ of the
metaplectic group, given by
\begin{equation}
\widehat{V}_{-P}\psi(x)=e^{\frac{i}{2\hbar}Px^{2}}\text{ \ },\text{
\ }\widehat{M}_{L,m}\psi(x)=i^{m}\sqrt{|\det L|}\psi(Lx) \label{vpmlhat}%
\end{equation}
where the integer $m$ corresponds to a choice of $\arg\det L$, and%
\begin{equation}
\widehat{J}\psi(x)=i^{-n/2}\widehat{\psi}(x)=\left(  \tfrac{1}{2\pi\hbar
i}\right)  ^{n/2}\int e^{-\frac{i}{\hbar}x\cdot x^{\prime}}\psi(x^{\prime
})d^{n}x^{\prime}~. \label{Jhat}%
\end{equation}

For a detailed discussion of the properties of $\operatorname*{Mp}(n)$ and its
generators see \cite{Birk,QUANTA}.

\subsection{The symplectic camel and related objects\label{secsymp}}

\subsubsection{Gromov's symplectic non-squeezing theorem}

In 1985 the mathematician M. Gromov \cite{gr85} proved the following
remarkable and highly non-trivial result: let $Z_{j}^{2n}(r)$ be the phase
space cylinder defined by $x_{j}^{2}+p_{j}^{2}\leq r^{2}$ and $\mathcal{B}%
^{2n}(R)$ the centered phase space ball with radius $R$. There exists a
canonical transformation $f$ of $\mathbb{R}_{z}^{2n}$ such that $f(\mathcal{B}%
^{2n}(R))\subset Z_{j}^{2n}(r)$ if and only $R\leq r$. This result (the
symplectic non-squeezing theorem) was reformulated by Gromov and Eliashberg
\cite{EliGro} in the following form: let $f$ be a canonical transformation of
$\mathbb{R}_{z}^{2n}$ and $\Pi_{j}$ the orthogonal projection $\mathbb{R}%
_{z}^{2n}\longrightarrow\mathbb{R}_{x_{j},p_{j}}^{2}$ on any plane of
conjugate variables $x_{j},p_{j}$. Then
\begin{equation}
\operatorname*{Area}\Pi_{j}(f(\mathcal{B}^{2n}(R)))\geq\pi R^{2}~.
\end{equation}
Of course the second result trivially implies the first, while the converse
implication follows from the fact that any planar domain of area smaller than
$\pi R^{2}$ can be mapped into a disk of the same area by an area-preserving
diffeomorphism. This result is called --- with a slight abuse of language ---
the \emph{principle of the symplectic camel}. We have used the latter in
\cite{go09} to reformulate the quantum uncertainty principle (see below),
using the related notion of \emph{symplectic capacity}. This principle
demonstrates that Gromov's theorem can be viewed as a watermark of quantum
mechanics in classical (Hamiltonian) mechanics; see the discussions in
\cite{golu09} and \cite{gohi}; in the latter \textquotedblleft the imprints of
the quantum world in classical mechanics\textquotedblright\ are discussed from
the point of view of symplectic topology. Also see the discussion in
\cite{Kalo} from the Hamiltonian point of view.

\subsubsection{Symplectic capacities}

For a detailed discussion of the notion of symplectic capacity and its
applications in physics see \cite{golu09}.

A (normalized) symplectic capacity on $(\mathbb{R}_{z}^{2n},\omega)$
associates to every subset $\Omega$ of $\mathbb{R}_{z}^{2n}$ a number
$c(\Omega)\in\mathbb{[}0,+\infty\mathbb{]}$ such that the following properties
hold \cite{ekhof1,ekhof2}:

\begin{itemize}
\item \textit{Monotonicity}: If $\Omega\subset\Omega^{\prime}$ then
$c(\Omega)\leq c(\Omega^{\prime})$;

\item \textit{Conformality}: For every real scalar $\lambda$ we have
$c(\lambda\Omega)=\lambda^{2}c(\Omega)$;

\item \textit{Symplectic invariance}: We have $c(f(\Omega))=c(\Omega)$ for
every canonical transformation $f\in\operatorname*{Symp}(n)$;

\item \textit{Normalization}: We have, for $1\leq j\leq n$,
\begin{equation}
c(\mathcal{B}^{2n}(R))=\pi R^{2}=c(Z_{j}^{2n}(R)) \label{cbz}%
\end{equation}
where $Z_{j}^{2n}(R)$ is the cylinder $\{(x,p):x_{j}^{2}+p_{j}^{2}\leq
R^{2}\}$.
\end{itemize}

Notice that the symplectic invariance of a symplectic capacity implies in
particular that
\begin{equation}
c(S(\Omega))=c(\Omega)\text{ \ if \ }S\in\operatorname*{Sp}(n)~.
\label{sympinv}%
\end{equation}

The symplectic capacities $c_{\min}$ and $c_{\max}$ are defined by%
\begin{subequations}
\begin{align}
c_{\min}(\Omega)  &  =\sup_{f\in\operatorname*{Symp}(n)}\{\pi R^{2}%
:f(\mathcal{B}^{2n}(R))\subset\Omega\}\label{cmin}\\
c_{\max}(\Omega)  &  =\inf_{f\in\operatorname*{Symp}(n)}\{\pi R^{2}%
:f(\Omega)\subset Z_{j}^{2n}(R)~. \label{cmax}%
\end{align}
That $c_{\min}$ and $c_{\max}$ indeed are symplectic capacities follows from
Gromov's symplectic non-squeezing theorem \cite{gr85}. Some terminology:
$c_{\min}$ is called the \textquotedblleft Gromov width\textquotedblright%
\ while $c_{\max}$ is the \textquotedblleft cylindrical
capacity\textquotedblright. This is because $c_{\max}$ measures the area of
the base of the smallest cylinder into which a subset of phase space can be
symplectically embedded. The notation $c_{\min}$ and $c_{\max}$ is motivated
by the fact that they are the smallest (\textit{resp.} the largest) symplectic
capacities: every symplectic capacity $c$ on $(\mathbb{R}_{z}^{2n},\omega)$
satisfies
\end{subequations}
\begin{equation}
c_{\min}(\Omega)\leq c(\Omega)\leq c_{\max}(\Omega) \label{cminmax}%
\end{equation}
for all $\Omega\subset\mathbb{R}_{z}^{2n}$.

It should be observed that for $n>1$ symplectic capacities are not related to
the notion of volume; the symplectic capacity of a set can be finite while
having infinite volume (this is the case of the cylinders $Z_{j}^{2n}(R)$).
Heuristically one can view a symplectic capacity as a generalization of the
notion of area, or (equivalently) of that of action. For instance, it is
possible to show that a particular symplectic capacity (the Hofer--Zehnder
capacity \cite{Polter}) of a compact convex set $\Omega$ with smooth boundary
$\partial\Omega$ is the action integral $\int_{\gamma}pdx$ calculated along
the shortest periodic orbit $\gamma$ carried by $\partial\Omega$
(\textquotedblleft Hofer--Zehnder capacity\textquotedblright).

\subsubsection{The symplectic capacity of an ellipsoid}

In what follows we use the notation $Mz^{2}=Mz\cdot z=z^{T}Mz$ ($M$ a square
matrix); $M>0$ means that $M$ is symmetric: $M=M^{T}$ and positive
definite,\textit{ i.e}. $Mz^{2}>0$ for all $z\neq0$. A remarkable property is
that all symplectic capacities agree on ellipsoids: if
\[
\Omega=\{z\in\mathbb{R}_{z}^{2n}:Mz^{2}\leq R^{2}\}
\]
where $M>0$, then for every symplectic capacity $c$ on $(\mathbb{R}_{z}%
^{2n},\omega)$ we have
\begin{equation}
c(\Omega)=\pi R^{2}/\nu_{\max} \label{cell}%
\end{equation}
where $\nu_{\max}$ is the largest symplectic eigenvalue of $M$. (Recall that
the symplectic eigenvalues $\nu_{1},...,\nu_{n}$ of $M$ are the numbers
$\nu_{j}>0$ defined by the condition \textquotedblleft\ $\pm i\nu_{j}$ is an
eigenvalue of $JM$\textquotedblright.) This property allowed us to prove in
\cite{go09} that the RSUP is equivalent to the inequality%
\begin{equation}
c(\Omega)\geq\pi\hbar\label{foop}%
\end{equation}
when $\Omega$ is a quantum covariance ellipsoid. From this formula the
symplectic invariance of the RSUP becomes obvious since we have $c(S(\Omega
))=c(\Omega)$ for every $S\in\operatorname*{Sp}(n)$.

\section{Quantum Dual Pairs and Covariance Ellipsoids}

\subsection{Quantum polar duality}

Polar duality is a very useful mathematical tool in convex and asymptotic
geometry, and in functional analysis. It has also recently been used by
Kalogeropoulos \cite{Kalo} to discuss phase space coarse-graining.

\subsubsection{Polar duality in convex geometry}

Let $X$ be a convex body in configuration space $\mathbb{R}_{x}^{n}$ (a convex
body in an Euclidean space is a compact convex set with non-empty interior).
We assume in addition that $X$ contains $0$ in its interior. This is the case
if, for instance, $X$ is symmetric: $X=-X$.

\begin{definition}
The \emph{polar dual} of $X$ is the subset
\begin{equation}
X^{\hbar}=\{p\in\mathbb{R}_{p}^{n}:px\leq\hbar\text{ \textit{for all} }x\in
X\} \label{omo}%
\end{equation}
of the dual space $\mathbb{R}_{p}^{n}\equiv(\mathbb{R}_{x}^{n})^{\ast}$.
\end{definition}

Notice that it trivially follows from the definition that $X^{\hbar}$ is
convex. In the mathematical literature one usually chooses $\hbar=1$, in which
case one writes $X^{o}$ for the polar dual; we have $X^{\hbar}=\hbar X^{o}$.
Here is an intuitive interpretation of the polar dual: $X$ being convex it is
the intersection of a (possibly infinite) family of half spaces (the
\textquotedblleft supporting hyperplanes\textquotedblright\ of $X$).
Therefore, the polar of $X$ can be seen as the convex hull of a (possibly
infinite) set of points, coming from all of the supporting hyperplanes.

The following properties of the polar dual are obvious:
\begin{equation}
\text{\textit{Biduality}: }(X^{\hbar})^{\hbar}=X~; \label{biduality}%
\end{equation}%
\begin{equation}
\text{\textit{Antimonotonicity: }}X\subset Y\Longrightarrow Y^{\hbar}\subset
X^{\hbar}~; \label{antimonotonicity}%
\end{equation}%
\begin{equation}
\text{\textit{Scaling}: }\det L\neq0\Longrightarrow(LX)^{\hbar}=(L^{T}%
)^{-1}X^{\hbar}~. \label{scaling}%
\end{equation}

The \textquotedblleft smaller\textquotedblright\ $X$ is, the larger $X^{\hbar
}$ is. For instance, if $X=0$ (corresponding to a perfectly localized system)
then $X^{\hbar}=\mathbb{R}_{p}^{n}$, the whole momentum space. This property,
reminiscent of the uncertainty principle, and of the duality of the support of
a function and that of its Fourier transform, becomes particularly visible
when one studies the polar duals of ellipsoids. Here are a few useful results:

\begin{lemma}
Let $\mathcal{B}_{X}^{n}(R)$ (\textit{resp}. $\mathcal{B}_{P}^{n}(R)$) be the
ball $\{x:|x|\leq R\}$ in $\mathbb{R}_{x}^{n}$ (\textit{resp}. $\{p:|p|\leq
R\}$ in $\mathbb{R}_{p}^{n}$). (i) We have
\begin{equation}
\mathcal{B}_{X}^{n}(R)^{\hbar}=\mathcal{B}_{P}^{n}(\hbar/R)~. \label{BhR}%
\end{equation}
In particular
\begin{equation}
\mathcal{B}_{X}^{n}(\sqrt{\hbar})^{\hbar}=\mathcal{B}_{P}^{n}(\sqrt{\hbar})~.
\label{bhh}%
\end{equation}
(ii) Let $A=A^{T}$ be an invertible $n\times n$ matrix. We have
\begin{equation}
\{x:Ax^{2}\leq R^{2}\}^{\hbar}=\{p:A^{-1}p^{2}\leq(\hbar/R)^{2}\}
\label{dualell}%
\end{equation}
and hence%
\begin{equation}
\{x:Ax^{2}\leq\hbar\}^{\hbar}=\{p:A^{-1}p^{2}\leq\hbar\}~. \label{dualellh}%
\end{equation}

\end{lemma}

\begin{proof}
Let us show that $\mathcal{B}_{X}^{n}(R)^{\hbar}\subset\mathcal{B}_{P}%
^{n}(\hbar/R)$. Let $p\in\mathcal{B}_{X}^{n}(R)^{\hbar}$ and set $x=(R/|p|)p$;
we have $|x|=R$ and hence $px\leq\hbar$, that is $R|p|\leq\hbar$ and
$p\in\mathcal{B}_{P}^{n}(\hbar/R)$. To prove the opposite inclusion choose
$p\in\mathcal{B}_{P}^{n}(\hbar/R)$. We have $|p|\leq\hbar/R$ and hence, by the
Cauchy--Schwarz inequality, $px\leq|x||p|\leq\hbar|x|/R$, that is
$px\leq\hslash$ for all $x$ such that $|x|\leq R$; this means that
$p\in\mathcal{B}_{X}^{n}(R)^{\hbar}$. \textit{(ii)} The ellipsoid
$\{x:Ax^{2}\leq R^{2}\}^{\hbar}$ is the image of $\mathcal{B}_{X}^{n}(R)$ by
the automorphism $A^{-1/2}$; in view of formula (\ref{dualell}) it follows
from the scaling property (\ref{scaling}) and (\ref{BhR}) that
\[
\{x:Ax^{2}\leq R^{2}\}^{\hbar}=A^{1/2}\mathcal{B}_{X}^{n}(R)^{\hslash}%
=A^{1/2}\mathcal{B}_{P}^{n}(\hbar/R)
\]
which is equivalent to (\ref{dualell}).
\end{proof}

So far we have assumed that the convex body $X$ contains the origin $0$ in its
interior. The definitions and results listed above extend without difficulty
to the general case by picking an arbitrary $x_{0}\in X$ and replacing $X$
with $X_{0}=-x_{0}+X$.

\subsubsection{Quantum dual pairs}

The following definition will be motivated by Theorem \ref{Thm1} below:

\begin{definition}
A pair $(X,P)\ $of symmetric convex bodies $X\subset\mathbb{R}_{x}^{n}$ and
$P\subset\mathbb{R}_{p}^{n}$ is called a \textquotedblleft quantum dual
pair\textquotedblright\ (or, for short, \textquotedblleft dual
pair\textquotedblright) if we have $X^{\hbar}\subset P$ or, equivalently,
$P^{\hbar}\subset X$. When equality occurs, that is if $X^{\hbar}=P$ we say
that the dual pair $(X,P)$ is saturated.
\end{definition}

\begin{remark}
We want to make the reader aware that we are committing a slight abuse of
notation and terminology here, but this abuse helps us to avoid cluttering
notation and making statements unnecessarily complicated. Rigorously speaking,
the set $X$ is a subset the configuration space $\mathbb{R}_{x}^{n}$ of a
system, while $P$ is a subset of the momentum space $\mathbb{R}_{p}^{n}$ of
that system, which is algebraically and topologically the dual of
$\mathbb{R}_{x}^{n}$. This amounts to identify the phase space of the system
with the cotangent bundle $T^{\ast}\mathbb{R}_{x}^{n}=\mathbb{R}_{x}^{n}%
\times(\mathbb{R}_{x}^{n})^{\ast}$. However, since we will be working in
\textquotedblleft flat\textquotedblright\ configuration space we are
identifying $\mathbb{R}_{p}^{n}$ with a copy of $\mathbb{R}_{x}^{n}$ and the
phase space with the product $\mathbb{R}_{x}^{n}\times\mathbb{R}_{p}^{n}%
\equiv\mathbb{R}_{x,p}^{2n}\equiv\mathbb{R}_{z}^{2n}$.
\end{remark}

Here are two elementary but important properties of quantum pairs:%

\begin{gather}
\text{\textit{Let} }(X,P)\text{ \textit{be a quantum dual pair and} }Y,Q\text{
\textit{be symmetric} \textit{convex }}\label{XY}\\
\text{\textit{bodies such that }}X\subset Y\text{ and }P\subset
Q\text{\textit{. Then} }(Y,Q)\text{ \textit{is also a quantum dual pair.}%
}\nonumber
\end{gather}

\noindent This follows from the antimonotonicity of the passage to the dual
where we have the chain of inclusions $Y^{\hbar}\subset X^{\hbar}\subset
P\subset Q$;%

\begin{gather}
\text{\textit{Two ellipsoids} }X=\{x:Ax^{2}\leq\hbar\}\text{ \textit{and }%
}P=\{p:Bp^{2}\leq\hbar\}\text{ }\label{Dual}\\
\text{\textit{form a quantum dual pair if and only if }}AB\leq I_{n\times
n}\nonumber\\
\text{ \textit{and} \textit{we have the equality}.}X^{\hbar}=P\text{
\textit{if and only if} }AB=I_{n\times n}~.\nonumber
\end{gather}

\noindent This follows from the slightly more general statement: if
$X=\{x:Ax^{2}\leq R^{2}\}$ and $P=\{p:Bp^{2}\leq R^{\prime2}\}$ then $(X,P)$
is a quantum dual pair if and only if $AB\leq(\hbar^{-1}R^{\prime}R)^{2}$. In
view of the duality formula (\ref{dualell}) we have
\[
X^{\hbar}=\{p:A^{-1}p^{2}\leq(\hbar/R)^{2}\}
\]
and we thus have $X^{\hbar}\subset P$ if and only if $R^{2}\hbar^{-2}%
A^{-1}x^{2}\leq1$ implies $(R^{\prime})^{-2}Bx^{2}\leq1$. But this condition
is in turn equivalent to $R^{2}\hbar^{-2}A^{-1}\geq(R^{\prime})^{-2}B$, that
is $AB\leq(\hbar^{-1}R^{\prime}R)^{2}$. We have used here the following
property of the L\"{o}wner ordering: if $K$ and $L$ are positive definite
symmetric matrices such that $K^{-1}\geq L$ then $KL\leq I_{n\times n}$ (and
conversely): $K^{-1}\geq L$ is equivalent to $K^{1/2}LK^{1/2}\leq I_{n\times
n}$ and $K^{1/2}LK^{1/2}$ and $KL$ have the same eigenvalues.

\subsubsection{Polar duality and Lagrangian planes}

We have defined polar duality in terms of the subspaces $\mathbb{R}_{x}%
^{n}\equiv\mathbb{R}_{x}^{n}\times0$ and $\mathbb{R}_{p}^{n}\equiv
0\times\mathbb{R}_{p}^{n}$ of the phase space $\mathbb{R}_{z}^{2n}$. These
subspaces have the property that the symplectic form $\omega$ vanishes
identically on them: $\omega(x,0;x^{\prime},0)=0$ and $\omega(0,p;0.p^{\prime
})=0$ for all $x^{\prime}$ and $p^{\prime}$. Any $n$-dimensional subspace
$\ell$ of $\mathbb{R}_{z}^{2n}$ on which $\omega$ is identically zero is
called a \textit{Lagrangian plane }in the symplectic literature \cite{Birk}.
The set of all Lagrangian planes in $(\mathbb{R}_{z}^{2n},\omega)$ is denoted
by $\operatorname*{Lag}(n)$ (or sometimes $\Lambda(n)$) and is called the
\textit{Lagrangian Grassmannian} of\textit{ }$(\mathbb{R}_{z}^{2n},\omega)$.
One can show that $\operatorname*{Lag}(n)$ can be identified with the
homogeneous space $U(n)/O(n)$ and equipped with its natural topology. Since
the image of a Lagrangian plane by a symplectic transformation obviously also
is a Lagrangian plane, it follows that we have a continuous and transitive
action%
\[
\operatorname*{Sp}(n)\times\operatorname*{Lag}(n)\longrightarrow
\operatorname*{Lag}(n)~.
\]

It turns out that we can define the notion of polarity for any pair
$(\ell,\ell^{\prime})$ of transversal Lagrangian planes, that is, $\ell
\cap\ell^{\prime}=0$. To see this we begin by making the following remark
which relates the notion of polarity to the symplectic structure: let $X$ be,
as before, a convex body in $\mathbb{R}_{x}^{n}\equiv\mathbb{R}_{x}^{n}%
\times0$ containing the origin. We observe that the polar dual $X^{\hbar}$ is
the subset of $\mathbb{R}_{p}^{n}\equiv0\times\mathbb{R}_{p}^{n}$ defined by
the condition
\[
z\in X^{\hbar}\Longleftrightarrow\omega(z,z^{\prime})\leq\hbar\text{ for all
}z^{\prime}\in X~\text{.}%
\]
Indeed, since $z=(0,p)$ and $z^{\prime}=(x^{\prime},0)$ for some $p$ and
$x^{\prime}$ we have $\omega(z,z^{\prime})=\omega((0,p;x^{\prime
},0)=px^{\prime}$whence $z\in X^{\hbar}$ means that $px^{\prime}\leq\hbar$,
which is the usual definition (\ref{omo}) of the polar dual. This observation
motivates the following definition:

\begin{definition}
\label{defl1}Let $(\ell,\ell^{\prime})$ be a pair of transversal Lagrangian
planes, and $X_{\ell}$ a convex body containing the origin in $\ell$. The
polar dual $X_{\ell^{\prime}}^{\hbar}$ of $X_{\ell}$ is the subset of
$\ell^{\prime}$ consisting of all $z^{\prime}\in\ell^{\prime}$ such that
\begin{equation}
\omega(z,z^{\prime})\leq\hbar\text{ \ for all \ }z\in X_{\ell}~. \label{ozz}%
\end{equation}

\end{definition}

This\ definition and the discussion that precedes it show that the notion of
polar dual is actually of a deep symplectic nature. It turns out that using
the properties of Lagrangian planes (see \cite{Birk} for a detailed study) one
can give the following equivalent definition of polar duality:

\begin{definition}
\label{defl2}Let $(\ell,\ell^{\prime})$ be a pair of transversal Lagrangian
planes, and set $\ell_{X}=\mathbb{R}_{x}^{n}\times0$, $\ell_{P}=0\times
\mathbb{R}_{p}^{n}$. Let $S\in\operatorname*{Sp}(n)$ be such that $\ell
=S\ell_{X}$ and $\ell^{\prime}=S\ell_{P}$. and $X_{\ell}\subset\ell$ a convex
body containing the origin. The polar dual $X_{\ell^{\prime}}^{\hbar}%
\subset\ell^{\prime}$ of $X_{\ell}$ is defined by
\begin{equation}
X_{\ell^{\prime}}^{\hbar}=S(S^{-1}X_{\ell})^{\hbar} \label{sxl}%
\end{equation}
\ where $(S^{-1}X_{\ell})^{\hbar}$ is the polar dual of $S^{-1}X_{\ell}%
\subset\ell_{X}$.
\end{definition}

The key point in the proof of the equivalence of both definitions lies in the
following property of the Lagrangian Grassmannian $\operatorname*{Lag}(n)$:
the symplectic group $\operatorname*{Sp}(n)$ acts transitively on pairs of
transverse Lagrangian planes; \textit{i.e}. if $(\ell_{1},\ell_{1}^{\prime})$
and $(\ell_{2},\ell_{2}^{\prime})$ are pairs of elements of
$\operatorname*{Lag}(n)$ such that $\ell_{1}\cap\ell_{1}^{\prime}=\ell_{2}%
\cap\ell_{2}^{\prime}=0$ then there exists exactly one $S\in\operatorname*{Sp}%
(n)$ such that $(\ell_{1},\ell_{1}^{\prime})=S(\ell_{2},\ell_{2}^{\prime
})=(S\ell_{2},S\ell_{2}^{\prime})$. Specializing to the case $(\ell_{1}%
,\ell_{1}^{\prime})=(\ell,\ell^{\prime})$ and $(\ell_{2},\ell_{2}^{\prime
})=(\ell_{X},\ell_{P})$ the existence of $S\in\operatorname*{Sp}(n)$ in
Definition \ref{defl2} follows. There remains to show that conditions
(\ref{ozz}) and (\ref{sxl}) are indeed equivalent. Let $z^{\prime}\in
X_{\ell^{\prime}}^{\hbar}$, equivalently $S^{-1}z^{\prime}=(S^{-1}X_{\ell
})^{\hbar}\in\ell_{P}$ which means that $\omega(S^{-1}z,S^{-1}z^{\prime}%
)\leq\hbar$ for all $S^{-1}z\in\ell_{X}$. Since $\omega(S^{-1}z,S^{-1}%
z^{\prime})=\omega(z,z^{\prime})$ we thus have $\omega(z,z^{\prime})\leq\hbar$
for all $z\in X_{\ell}$ which means that $X_{\ell^{\prime}}^{\hbar}$ is the
polar dual of $X_{\ell}$ by definition (\ref{ozz}). Reversing the argument we
see that, conversely, (\ref{ozz}) implies (\ref{sxl}).

The properties of this generalized notion of polar duality are similar to
those in the standard case. For instance, we have the biduality property
\[
(X_{\ell^{\prime}}^{\hbar})_{\ell}^{\hbar}=X_{\ell}%
\]
which readily follows from definition (\ref{sxl}), swapping the roles of
$\ell$ and $\ell^{\prime}$:
\begin{align*}
(X_{\ell^{\prime}}^{\hbar})_{\ell}^{\hbar}  &  =S[S^{-1}X_{\ell^{\prime}%
}^{\hbar}]^{\hbar}=S[S^{-1}S(S^{-1}X_{\ell})^{\hbar}]^{\hbar}\\
&  =S[(S^{-1}X_{\ell})^{\hbar}]^{\hbar}=X_{\ell}~.
\end{align*}
Similarly, the antimonotonicity property (\ref{antimonotonicity}) becomes
\begin{equation}
X_{\ell}\subset Y_{\ell}\Longrightarrow Y_{\ell^{\prime}}^{\hbar}\subset
X_{\ell^{\prime}}^{\hbar}~. \label{antibis}%
\end{equation}
Let us see what happens with the scaling property (\ref{scaling}); what
follows will shed more light on its meaning which was hidden because of our
identification of $\ell_{X}=\mathbb{R}_{x}^{n}\times0$ with $\mathbb{R}%
_{x}^{n}$. Formula (\ref{scaling}) says that for every automorphism $L$ of
$\mathbb{R}_{x}^{n}$ (\textit{i.e.} an invertible $n\times n$ matrix) we have
$(LX)^{\hbar}=(L^{T})^{-1}X^{\hbar}$. However, if we view $X$ as a subset of
$\mathbb{R}_{x}^{n}\times0$ it must be acted upon by automorphisms of
$\mathbb{R}_{z}^{2n}$ (\textit{i.e.} $2n\times2n$ matrices). Using the
symplectic matrix $M_{L^{-1}}=%
\begin{pmatrix}
L & 0\\
0 & (L^{-1})^{T}%
\end{pmatrix}
$ we can view $LX$ as $M_{L^{-1}}X$ when $X\subset\ell_{X}=\mathbb{R}_{x}%
^{n}\times0$. Similarly, $(L^{T})^{-1}X^{\hbar}$ can be viewed as $M_{L^{-1}%
}X^{\hbar}$ when $X^{\hbar}\subset\ell_{P}=0\times\mathbb{R}_{p}^{n}$. With
this notation the scaling formula now reads
\[
(M_{L^{-1}}X)^{\hbar}=M_{L^{-1}}X^{\hbar}~.
\]
It turns out that this is a particular case of the following general formulas:
for every $S_{0}\in\operatorname*{Sp}(n)$ and $X_{\ell}\subset\ell$ we have
\begin{equation}
(S_{0}X_{\ell})_{S_{0}\ell^{\prime}}^{\hbar}=S_{0}X_{\ell}^{\hbar} \label{sox}%
\end{equation}
as follows from the observation that $S_{0}$ takes $\ell$ to $S_{0}\ell$ and
$\ell^{\prime}$ to $S_{0}\ell^{\prime}$.

The notion of dual quantum pair is defined accordingly: it is a pair
$(X_{\ell},P_{\ell^{\prime}})$ with $X_{\ell}\subset\ell$ and $P_{\ell
^{\prime}}\subset\ell^{\prime}$ such that $X_{\ell}^{\hbar}\subset
P_{\ell^{\prime}}$. Formula (\ref{sox}) shows that a symplectic transformation
$S_{0}\in\operatorname*{Sp}(n)$ takes a dual quantum pair $(X_{\ell}%
,P_{\ell^{\prime}})$ into the dual quantum pair $(S_{0}X_{\ell},S_{0}%
P_{\ell^{\prime}})$ since $S_{0}P_{\ell^{\prime}}\supset S_{0}X_{\ell}^{\hbar
}=(S_{0}X_{\ell})_{S_{0}\ell^{\prime}}^{\hbar}$. This is a geometric
generalization of the symplectic invariance of the Robertson--Schr\"{o}dinger
inequalities \ref{RS}.

\subsubsection{Polar duality and symplectic capacity}

Let us specialize the properties above to the case of ellipsoids. We first
recall the following symplectic result (see \cite{golu09}, \cite{Birkbis},
\S 6.2.1, or \cite{ACHA}, Lemma 6):

\begin{lemma}
\label{ABL}Let $A$ and $B$ be two real positive definite symmetric $n\times n$
matrices. There exists an invertible real $n\times n$ matrix $L$ such that
\begin{equation}
L^{T}AL=L^{-1}B(L^{T})^{-1}=\Lambda\label{lalb}%
\end{equation}
where $\Lambda=\operatorname*{diag}(\sqrt{\lambda_{1}},...,\sqrt{\lambda_{n}%
})$ the $\lambda_{j}>0$ being the eigenvalues of $AB$.
\end{lemma}

Observe that the eigenvalues $\lambda_{j}$ of $AB$ are $>0$ since they are the
same as those of $A^{1/2}BA^{1/2}$. This result may be viewed as a special
case of Williamson's symplectic diagonalization result; we can in fact rewrite
(\ref{lalb}) as
\begin{equation}%
\begin{pmatrix}
A & 0\\
0 & B
\end{pmatrix}
=%
\begin{pmatrix}
(L^{T})^{-1} & 0\\
0 & L
\end{pmatrix}%
\begin{pmatrix}
\Lambda & 0\\
0 & \Lambda
\end{pmatrix}%
\begin{pmatrix}
L^{-1} & 0\\
0 & L^{T}%
\end{pmatrix}
\label{willab}%
\end{equation}
and note that $S=%
\begin{pmatrix}
L^{-1} & 0\\
0 & L^{T}%
\end{pmatrix}
$ is symplectic and $S^{T}=%
\begin{pmatrix}
(L^{T})^{-1} & 0\\
0 & L
\end{pmatrix}
$. A very important property is that the cylindrical symplectic capacity
$c_{\max}$ of a dual quantum pair can be explicitly calculated. In fact, we
have the following generalization of the relation
\begin{equation}
\operatorname*{Area}(X\times X^{\hbar})=4\hbar\label{area}%
\end{equation}
for intervals:

\begin{theorem}
\label{Theorem1}(i) Let $(X,P)$ be an arbitrary pair of centrally symmetric
convex bodies $X\subset\mathbb{R}_{x}^{n}$ and $P\subset\mathbb{R}_{p}^{n}$;
we have
\begin{equation}
c_{\max}(X\times P)=4\hbar\max\{\lambda>0:\lambda X^{\hbar}\subset P\}~.
\label{yaron1}%
\end{equation}
(ii) Assume that $X=\{x:Ax^{2}\leq\hbar\}$ and $P=\{p:Bp^{2}\leq\hbar\}$ with
$A,B$ symmetric and positive definite, and $AB\leq I_{n\times n}$. We have
\begin{equation}
c_{\max}(X\times P)=4\hbar\max\nolimits_{j}\{\lambda_{j}^{-1}\} \label{cmaxj}%
\end{equation}
the $\lambda_{j}\leq1$ being the eigenvalues of $AB$ ($\lambda_{j}>0$). In
particular
\begin{equation}
c_{\max}(X\times P)\geq4\hbar\label{yaron}%
\end{equation}
with equality if and only if $P=X^{\hbar}$.
\end{theorem}

\begin{proof}
In \cite{arkaos13} (Remark 4.2) Artstein-Avidan \textit{et} \textit{al.} show
that (\ref{yaron1}) holds for $\hbar=1$; an elementary rescaling argument
immediately yields the general case. Let now $X=\{x:Ax^{2}\leq\hbar\}$ and
$P=\{p:Bp^{2}\leq\hbar\}$ and choose $L$ such that $L^{T}AL=L^{-1}%
B(L^{T})^{-1}=\Lambda$ (Lemma \ref{ABL}). We have
\begin{align*}
L^{-1}(X)  &  =\{x:%
{\textstyle\sum_{j=1}^{n}}
\lambda_{j}^{1/2}x_{j}^{2}\leq\hbar\}=\Lambda^{-1/4}\mathcal{B}_{X}^{n}%
(\sqrt{\hbar})\\
L^{T}(P)  &  =\{p:%
{\textstyle\sum_{j=1}^{n}}
\lambda_{j}^{1/2}p_{j}^{2}\leq\hbar\}=\Lambda^{-1/4}\mathcal{B}_{P}^{n}%
(\sqrt{\hbar})~
\end{align*}
and thus
\begin{align*}
c_{\max}(X\times P)  &  =c_{\max}(L^{-1}(X)\times L^{T}(P))\\
&  =c_{\max}(\Lambda^{-1/4}\mathcal{B}_{X}^{n}(\sqrt{\hbar})\times
\Lambda^{-1/4}\mathcal{B}_{P}^{n}(\sqrt{\hbar}))
\end{align*}
where the first equality follows from the symplectic invariance formula
(\ref{sympinv}). To prove (\ref{cmaxj}) let us determine the largest
$\lambda>0$ such that
\begin{multline*}
c_{\max}(\Lambda^{-1/4}\mathcal{B}_{X}^{n}(\sqrt{\hbar})\times\Lambda
^{-1/4}\mathcal{B}_{P}^{n}(\sqrt{\hbar}))=\\
\lambda(\Lambda^{-1/4}\mathcal{B}_{X}^{n}(\sqrt{\hbar}))^{\hbar}\subset
\Lambda^{-1/4}\mathcal{B}_{P}^{n}(\sqrt{\hbar})~.
\end{multline*}
We have $(\Lambda^{-1/4}\mathcal{B}_{X}^{n}(\sqrt{\hbar}))^{\hbar}%
=\Lambda^{1/4}\mathcal{B}_{P}^{n}(\sqrt{\hbar})$ and hence
\[
\lambda\Lambda^{1/4}\mathcal{B}_{X}^{n}(\sqrt{\hbar}))\subset\Lambda
^{-1/4}\mathcal{B}_{P}^{n}(\sqrt{\hbar})
\]
or, equivalently, $\lambda\mathcal{B}_{X}^{n}(\sqrt{\hbar})\subset
\Lambda^{-1/2}\mathcal{B}_{P}^{n}(\sqrt{\hbar})$. But this means that we must
have $\lambda^{2}\geq\lambda_{j}^{-1}$ for all $j=1,...,n$. This proves
formula (\ref{cmaxj}). The inequality (\ref{yaron}) follows since $AB\leq
I_{n\times n}$. Suppose that $P=X^{\hbar}$; then $B=A^{-1}$ so that the
eigenvalues $\lambda_{j}$ are all equal to one, hence $c_{\max}(X\times
X^{\hbar})=4\hbar$. If conversely $c_{\max}(X\times P)=4\hbar$ then we must
have $\max\nolimits_{j}\{\lambda_{j}^{-1}\}=1$ that is again $\lambda_{j}=1$
for all $j$ which means that we have $A^{1/2}BA^{1/2}=I_{n\times n}$ and hence
$B=A^{-1}$, that is $P=X^{\hbar}$.
\end{proof}

\subsection{Polar duality by orthogonal projections}

\subsubsection{The Schur complement}

It will be convenient to introduce the matrix
\begin{equation}
M=\frac{\hbar}{2}\Sigma^{-1} \label{MH}%
\end{equation}
in which case the covariance ellipsoid takes the form
\begin{equation}
\Omega=\{z:Mz\cdot z\leq\hbar\}~. \label{covell1}%
\end{equation}
The matrix $M$ is a real positive definite symmetric $2n\times2n$ matrix:
$M=M^{T}>0$. We will write it in block-matrix form
\begin{equation}
M=%
\begin{pmatrix}
M_{XX} & M_{XP}\\
M_{PX} & M_{PP}%
\end{pmatrix}
\label{M}%
\end{equation}
where the blocks are $n\times n$ matrices. The condition $M>0$ ensures us that
$M_{XX}>0$, $M_{PP}>0$, and $M_{PX}=M_{XP}^{T}$. Recall \cite{zhang} the
following definition: the $n\times n$ matrices
\begin{align}
M/M_{PP}  &  =M_{XX}-M_{XP}M_{PP}^{-1}M_{PX}\label{schurm1}\\
M/M_{XX}  &  =M_{PP}-M_{PX}M_{XX}^{-1}M_{XP} \label{schurm2}%
\end{align}
are the \textit{Schur complements} in $M$ of $M_{PP}$ and $M_{XX}$,
respectively, and we have $M/M_{PP}>0$, $M/M_{XX}>0$ \cite{zhang}.

We assume that the uncertainty principle in its form (\ref{Quantum}) holds.
This is equivalent to the existence of $S\in\operatorname*{Sp}(n)$ such that
$S(\mathcal{B}^{2n}(\sqrt{\hbar})\subset\Omega$ (see
\cite{go09,blob,Birk,golu09}).

Let $\Pi_{X}$ (resp. $\Pi_{P}$) be the orthogonal projection $\mathbb{R}%
_{z}^{2n}\longrightarrow\mathbb{R}_{x}^{n}$ (\textit{resp}. $\mathbb{R}%
_{z}^{2n}\longrightarrow\mathbb{R}_{p}^{n}$) and set
\begin{equation}
\Omega_{X}=\Pi_{X}\Omega\text{ \ },\text{ \ }\Omega_{P}=\Pi_{P}\Omega~.
\label{shadows}%
\end{equation}

\begin{lemma}
\label{LemmaProj}Let $\Omega=\{z:Mz\cdot z\leq\hbar\}$, $M>0$. The orthogonal
projections $\Omega_{X}$ and $\Omega_{P}$ of $\Omega$ are the ellipsoids%
\begin{align}
\Omega_{X}  &  =\{x\in\mathbb{R}_{x}^{n}:(M/M_{PP})x^{2}\leq\hbar
\}\label{boundb}\\
\Omega_{P}  &  =\{p\in\mathbb{R}_{p}^{n}:(M/M_{XX})p^{2}\leq\hbar\}~.
\label{bounda}%
\end{align}

\end{lemma}

\begin{proof}
Let us set $Q(z)=Mz^{2}-\hbar$; the boundary $\partial\Omega$ of the
hypersurface $Q(z)=0$ is defined by%
\begin{equation}
M_{XX}x^{2}+2M_{PX}x\cdot p+M_{PP}p^{2}=\hbar~. \label{mabab}%
\end{equation}
A point $x$ belongs to the boundary $\partial\Omega_{X}$ of $\Omega_{X}$ if
and only if the normal vector to $\partial\Omega$ at the point $z=(x,p)$ is
parallel to $\mathbb{R}_{x}^{n}\times0$ hence we get the constraint
$\nabla_{z}Q(z)=2Mz\in\mathbb{R}_{x}^{n}\times0$; this is equivalent to saying
that $M_{PX}x+M_{PP}p=0$, that is to $p=-M_{PP}^{-1}M_{PX}x$. Inserting this
value of $p$ in the equation (\ref{mabab}) shows that $\partial\Omega_{X}$ is
the set of all $x$ such that $(M/M_{PP})x^{2}=\hbar$, which yields
(\ref{boundb}). Formula (\ref{bounda}) is proven in the same way, swapping the
subscripts $X$ and $P$.
\end{proof}

\subsubsection{Proof of the projection theorem}

Let us now prove the projection theorem (we have given a proof thereof in
\cite{gopolar} when the covariance matrix is block-diagonal).

\begin{theorem}
\label{Thm1}Assume that the covariance ellipsoid $\Omega$ satisfies the
quantization condition $\Sigma+\frac{i\hbar}{2}J\geq0$ (resp. $M^{-1}+iJ\geq
0$). Then, the orthogonal projections $\Omega_{X}=\Pi_{X}\Omega$ and
$\Omega_{P}=\Pi_{P}\Omega$ of $\Omega$ on $\mathbb{R}_{x}^{n}$ and
$\mathbb{R}_{p}^{n}$, respectively, form a dual quantum pair: we have
$\Omega_{X}^{\hbar}\subset\Omega_{P}$.
\end{theorem}

\begin{proof}
In view of property (\ref{XY}) it suffices to prove that there exist
$Y\subset\Omega_{X}$ and $Q\subset\Omega_{P}$ such that $Y^{\hbar}\subset Q$.
For this purpose we recall (property (\ref{camel2})) that the quantum
condition $\Sigma+(i\hbar/2)J\geq0$ is equivalent to the existence of
$S\in\operatorname*{Sp}(n)$ such that $S(\mathcal{B}^{2n}(\sqrt{\hbar
}))\subset\Omega$; it is therefore sufficient to show that the projections
$Y=\Pi_{X}(S(\mathcal{B}^{2n}(\sqrt{\hbar}))$ and $Q=\Pi_{P}(S(\mathcal{B}%
^{2n}(\sqrt{\hbar}))$ form a quantum pair, that is $Y^{\hbar}\subset Q$. The
ellipsoid $S(\mathcal{B}^{2n}(\sqrt{\hbar}))$ consists of all $z\in
\mathbb{R}_{z}^{2n}$ such that $Rz^{2}\leq\hbar$ where $R=(SS^{T})^{-1}$.
Since $R$ is symmetric and positive definite we can write it in block-matrix
form as
\[
R=%
\begin{pmatrix}
A & B\\
B^{T} & D
\end{pmatrix}
\]
\ with $A>0$, $D>0$, and the projections $Y$ and $Q$ are given by formulas
(\ref{boundb}) and (\ref{bounda}), which read here
\begin{align*}
Y  &  =\{x:(R/D)x^{2}\leq\hbar\}\\
Q  &  =\{p:(R/A)p^{2}\leq\hbar\}~;
\end{align*}
we have $R/D>0$ and $R/A>0$ \cite{zhang}. In view of formula (\ref{dualellh})
we have
\[
Y^{\hbar}=\{p:(R/D)^{-1}p^{2}\leq\hbar\}
\]
hence the condition $Y^{\hbar}\subset Q$ is equivalent to%
\begin{equation}
(R/D)^{-1}\geq R/A~. \label{RPPXX}%
\end{equation}
Let us prove that this inequality holds. The conditions $R\in
\operatorname*{Sp}(n)$, $R=R^{T}$ being equivalent to $RJR=J$ we have%
\begin{gather}
AB^{T}=BA\text{ , \ }B^{T}D=DB\label{rxx}\\
AD-B^{2}=I_{n\times n}~. \label{rpp}%
\end{gather}
These relations\ imply that the Schur complements $R/D$ and $R/A$ are
\begin{align}
R/D  &  =(AD-B^{2})D^{-1}=D^{-1}\label{34}\\
R/A  &  =A^{-1}(AD-B^{2})=A^{-1} \label{35}%
\end{align}
and hence the inequality (\ref{RPPXX}) holds if and only $D\geq A^{-1}$. This
condition is in turn equivalent to $AD\geq I_{n\times n}$. In fact, the
inequality $D\geq A^{-1}$ is equivalent to $A^{1/2}DA^{1/2}\geq I_{n\times n}%
$; now $A^{1/2}DA^{1/2}$ and $AD$ have the same eigenvalues hence $AD\geq
I_{n\times n}$. If conversely $AD\geq I_{n\times n}$ then $D^{1/2}AD^{1/2}\geq
I_{n\times n}$ hence $A\geq D^{-1}$ that is $D\geq A^{-1}$. Now (\ref{rpp})
implies that $AD=I_{n\times n}+B^{2}$ hence we will have $AD\geq I_{n\times
n}$ if $B^{2}\geq0$. To prove that $B^{2}\geq0$ we note that since $AB^{T}=BA$
(first formula (\ref{rxx})) we have $B^{T}=A^{-1}BA$ so that $B$ and
$B^{T}=B^{\ast}$ have the same eigenvalues and these must be real. It follows
that the eigenvalues of $B^{2}$ are $\geq0$ hence $B^{2}\geq0$ as claimed.
\end{proof}

We will analyze the inverse problem in Section \ref{secpauli}, where we
investigate whether a covariance ellipsoid $\Omega$ can be reconstructed from
its orthogonal projections $\Omega_{X}$ and $\Omega_{P}$ (\textquotedblleft
Pauli's problem\textquotedblright).

\subsection{Quantum polarity and dynamics}

So far we have been dealing only with time-independent processes. Let us now
have a look at quantum polarity from a dynamical point of view.

\subsubsection{Quadratic Hamiltonians}

The following is well-known \cite{Birk,Birkbis,Littlejohn}. Let $H$ be a
Hamiltonian function on phase space $\mathbb{R}_{z}^{2n}$; we assume that $H$
is a quadratic form in the position and momentum variables $x_{j},p_{k}$. Such
a function can always be written as
\begin{equation}
H(z)=\frac{1}{2}H^{\prime\prime}z\cdot z \label{Hamquad}%
\end{equation}
where $H^{\prime\prime}$ (the Hessian of $H$) is a real symmetric $2n\times2n$
matrix ; it is convenient to rewrite this as%
\[
H(z)=\frac{1}{2}JXz\cdot z
\]
where $X=-JH^{\prime\prime}$ satisfies the condition $JX+X^{T}J=0$,
\textit{i.e.} $X$ is in the symplectic Lie algebra $\mathfrak{sp}(n)$
\cite{Birk}. The associated Hamilton equations are $\dot{z}=JXz$ and its
solutions are hence $z(t)=e^{tJX}z(0)$. This means that the Hamiltonian flow
determined by the quadratic Hamiltonian function $H$ consists of the
symplectic matrices
\begin{equation}
S_{t}=e^{tJX}\in\operatorname*{Sp}(n)~. \label{stex}%
\end{equation}
Using the path-lifting theorem from the theory of fiber bundles, or performing
a direct (cumbersome) calculation one shows that the solution of the
corresponding Schr\"{o}dinger equation
\[
i\hbar\frac{\partial\psi}{\partial t}=\widehat{H}(x,-i\hbar\nabla_{x})\psi
\]
where $\widehat{H}=\widehat{H}(x,-i\hbar\nabla_{x})$ is the Weyl quantization
of $H$, is given by
\[
\psi(x,t)=\widehat{S}_{t}\psi(x,0)
\]
where the unitary operators $\widehat{S_{t}}\in\operatorname*{Mp}(n)$ are
defined as follows: as $t$ varies, the matrices $S_{t}$ describe a smooth path
in $\operatorname*{Sp}(n)$ passing through the identity $I_{n\times n}$ at
time $t=0$. To this path corresponds a unique smooth path of metaplectic
operators $\widehat{S}_{t}$ in $\operatorname*{Mp}(n)$ such that
$\widehat{S}_{0}$ is the identity operator, and this path is precisely the
quantum propagator, \textit{i.e}. $\widehat{S}_{t}=e^{-i\widehat{H}t/\hbar}$.

The discussion above extends to the case where the Hamiltonian has
time-dependent coefficients,\textit{ i.e.} is of the type
\[
H(z,t)=\frac{1}{2}H^{\prime\prime}(t)z\cdot z
\]
where $H^{\prime\prime}(t)$ depends continuously on $t$. In this case the
propagator $S_{t}$ cannot in general be written in a simple explicit form, and
it is advantageous to use the time-dependent flow defined by $S_{t,t^{\prime}%
}=S_{t}(S_{t^{\prime}})^{-1}$. The symplectic transformation $S_{t,t^{\prime}%
}$ takes a phase space point $z^{\prime}$ at time $t^{\prime}$ to a point $z$
at time $t$. In a similar way one can consider the Schr\"{o}dinger equation%
\[
i\hbar\frac{\partial\psi}{\partial t}=\widehat{H}(x,-i\hbar\nabla_{x},t)\psi
\]
whose solution is given by
\[
\psi(x,t)=\widehat{S}_{t,t^{\prime}}\psi(x,t^{\prime})
\]
where the $\widehat{S}_{t,t^{\prime}}\in\operatorname*{Mp}(n)$ are defined by
lifting the time-dependent flow $S_{t,t^{\prime}}$ to the metaplectic group;
see \cite{Birkbis} for details.

\subsubsection{Time-evolution of a dual quantum pair}

Let us begin by studying the orthogonal projections of a phase space ellipsoid
under the action of the flow $S_{t}\in\operatorname*{Sp}(n)$ (to simplify the
notation we are limiting ourselves here to the case of a time-independent
Hamiltonian; everything carries over to the time-dependent case without
difficulty). Consider an ellipsoid
\begin{equation}
\Omega=\{z:Mz^{2}\leq\hbar\}
\end{equation}
where as usual it is assumed that $M>0$. We have seen (Theorem \ref{Thm1})
that if $\Omega$ satisfies the quantum condition $M^{-1}+iJ\geq0$ (which we
assume from now on) then the projections $\Omega_{X}$ and $\Omega_{P}$ on the
coordinate Lagrangian planes $\ell_{X}=\mathbb{R}_{x}^{n}\times0$ and
$\ell_{P}=0\times\mathbb{R}_{x}^{n}$ form a dual quantum pair. Let $S_{t}%
\in\operatorname*{Sp}(n)$ be the flow determined by the Hamilton equations
associated with the quadratic Hamiltonian function (\ref{Hamquad}). As time
elapses, the ellipsoid $\Omega$ will deform into a new ellipsoid
\[
\Omega_{t}=S_{t}(\Omega)=\{z:M_{t}z^{2}\leq\hbar\}
\]
where $M_{t}=S_{-t}^{T}MS_{-t}$ (we have $S_{t}^{-1}=S_{-t}$). It is easily
seen that $M_{t}$ satisfies the quantum condition $M_{t}^{-1}+iJ\geq0$: since
$S_{t}$ is symplectic we have $S_{t}JS_{t}^{T}=J$ and hence%
\[
M_{t}^{-1}+iJ=S_{t}MS_{t}^{T}+iJ=S_{t}(M^{-1}+iJ)S_{t}^{T}\geq0~.
\]
It follows from Theorem \ref{Thm1} that the orthogonal projections
$\Omega_{t,X}$ and $\Omega_{t,P}$ again form a dual quantum pair; in fact
these projections can be calculated using formulas (\ref{boundb}) and
(\ref{bounda}):
\begin{align}
\Omega_{t,X}  &  =\{x:(M_{t}/M_{t,PP})x^{2}\leq\hbar\}\\
\Omega_{t,P}  &  =\{p:(M_{t}/M_{t,XX})p^{2}\leq\hbar\}
\end{align}
where $M_{t}/M_{t,PP}$ and $M_{t}/M_{t,XX}$ are the Schur complements in
$M_{t}$. We will not write the explicit formulas in terms of $M$ and $S_{t}$
here (they are quite complicated to work out), but rather focus on the volumes
of the corresponding ellipsoids. The Schur complement satisfies the identity
\cite{zhang}%
\[
\det M_{PP}\det(M/M_{PP})=\det M_{XX}\det(M/M_{XX})=\det M
\]
hence we will have, since $\det M_{t}=\det M$,
\begin{align*}
\det(M_{t}/M_{t,PP})  &  =\frac{\det M}{\det M_{t,PP}}\\
\det(M_{t}/M_{t,XX})  &  =\frac{\det M}{\det M_{t,XX}}~.
\end{align*}
It follows that
\begin{align*}
\operatorname*{Vol}\Omega_{t,X}  &  =\left(  \frac{1}{\det M_{t,PP}}\right)
^{n}\operatorname*{Vol}\Omega_{X}\\
\operatorname*{Vol}\Omega_{t,P}  &  =\left(  \frac{1}{\det M_{t,XX}}\right)
^{n}\operatorname*{Vol}\Omega_{P}%
\end{align*}
hence the volumes of the orthogonal projections $\Omega_{t,X}$ and
$\Omega_{t,P}$ are not constant in general, as opposed to the total volume
$\operatorname*{Vol}\Omega_{t}=\operatorname*{Vol}\Omega$ which is conserved
(Liouville's theorem). That we have a quantum-type spreading in the classical
case should not be surprising; such a possibility was already pointed out by
Littlejohn \cite{Littlejohn}. The reason behind this phenomenon lies in the
fact that there is a one-to-one correspondence between the classical flow of a
quadratic Hamiltonian and the corresponding quantum propagator, as discussed
above. So, the result above is just a quantum mechanical result in disguise.
In fact, suppose that
\[
\widehat{\rho}=\sum_{j}\lambda_{j}|\psi_{j}\rangle\langle\psi_{j}|
\]
is a quantum state with covariance matrix ellipsoid $\Omega$ and Wigner
distribution.%
\begin{equation}
W_{\widehat{\rho}}(z)=\sum_{j}\lambda_{j}W\psi_{j}(z)~.
\end{equation}
The time-evolution of $W_{\widehat{\rho}}(z)$ (and hence of $\widehat{\rho}$)
is given by the formula
\begin{align}
W_{\widehat{\rho}}(z,t)  &  =\sum_{j}\lambda_{j}W(\widehat{S}_{t}\psi
_{j})(z)\\
&  =\sum_{j}\lambda_{j}W(\psi_{j})(S_{-t}z)
\end{align}
from which it follows that the covariance ellipsoid of the evolved state
$\widehat{\rho}_{t}$ is precisely $S_{t}(\Omega)$.

\section{Pauli's Problem and Polar Duality\label{secpauli}}

Wave functions do not have an immediate experimental interpretation; what may
be deduced from experiments is rather the associated probability distributions
$|\psi(x)|^{2}$ and $|\widehat{\psi}(p)|^{2}$ (or, equivalently, the Wigner
transform $W\psi(x,p)$). Pauli asked in \cite{Pauli} the famous question
whether the probability densities $|\psi(x)|^{2}$ and $|\widehat{\psi}%
(p)|^{2}$ uniquely determine the wavefunction $\psi(x)$. in Pauli's words:

\begin{quotation}
\textquotedblleft\textit{The mathematical problem as to whether, for given
probability densities }$W(p)$\textit{ and }$W(x)$\textit{, the wavefunction
}$\psi$\textit{ (...) is always uniquely determined, has still not been
investigated in all its generality}\textquotedblright
\end{quotation}

We know that the answer is negative; in fact there is in general
non-uniqueness of the solution, which led Corbett \cite{Corbett} to introduce
the notion of \textquotedblleft Pauli partners\textquotedblright.
Mathematically speaking, the reconstruction problem we address here is that of
the reconstruction of a phase space ellipsoid (subject to a quantization
condition) from its orthogonal projections on the $x$- and $p$-spaces. It is a
particular case of what is called quantum tomography theory; see for instance
\cite{mancini,mankomanko,Moroz,pare}.

\subsection{The case $n=1$}

We have seen that the orthogonal projections $\Omega_{X}$ and $\Omega_{P}$ of
a quantum covariance matrix form a quantum dual pair. We now address the
converse question: if $(\Omega_{X},\Omega_{P})$ is a dual pair of ellipsoids,
is there a quantum covariance ellipsoid with orthogonal projections
$\Omega_{X}$ and $\Omega_{P}$ on $\mathbb{R}_{x}^{n}$ and $\mathbb{R}_{p}^{n}%
$? As follows from the discussion above, such a solution, if it exists, need
not be unique. Let us return to the dual pair of intervals $\Omega_{X}%
=[-\sqrt{2\sigma_{xx}},\sqrt{2\sigma_{xx}}]$\ and $\Omega_{P}=[-\sqrt
{2\sigma_{pp}},\sqrt{2\sigma_{pp}}\dot{]}$ considered in the introduction.
These intervals are the orthogonal projections on the $x$- and $p$-axes,
respectively, of \emph{any} covariance ellipse $\Omega$ defined by
\begin{equation}
\dfrac{\sigma_{pp}}{2D}x^{2}-\frac{\sigma_{xp}}{D}px+\dfrac{\sigma_{xx}}%
{2D}p^{2}\leq1 \label{2Dbis}%
\end{equation}
where $D=\sigma_{xx}\sigma_{pp}-\sigma_{xp}^{2}\geq\frac{1}{4}\hbar^{2}$
(formula (\ref{2D})). Knowledge of the variances $\sigma_{xx}$ and
$\sigma_{pp}$ does not suffice to determine uniquely $\Omega$, since we also
need to know the covariance $\sigma_{xp}^{2}$. We note, however, that every
ellipse (\ref{2Dbis}) has area%
\[
\operatorname*{Area}(\Omega)=2\pi\sqrt{D}\geq\pi\hbar~.
\]
This area condition thus excludes \textquotedblleft thin\textquotedblright%
\ ellipses concentrated along a diagonal of the rectangle $\Omega_{X}%
\times\Omega_{P}$. Suppose that the RSUP is saturated, that is, that
$D=\frac{1}{4}\hbar^{2}$. In this case $\operatorname*{Area}(\Omega)=\pi\hbar$
and the relation $\sigma_{xp}^{2}=\sigma_{xx}\sigma_{pp}-\frac{1}{4}\hbar^{2}$
determines $\sigma_{xp}$ \emph{up to a sign}: the state $\widehat{\rho}$ is
then either of the two pure Gaussians
\begin{equation}
\psi_{\pm}(x)=\left(  \tfrac{1}{2\pi\sigma_{xx}}\right)  ^{1/4}e^{-\frac
{x^{2}}{4\sigma_{xx}}}e^{\pm\frac{i\sigma_{xp}}{2\hbar\sigma_{xx}}x^{2}}
\label{sigmagauss1}%
\end{equation}
whose Fourier transforms are (up to an unimportant constant phase factor with
modulus one)%
\begin{equation}
\widehat{\psi}_{\pm}(p)=\left(  \tfrac{1}{2\pi\sigma_{pp}}\right)
^{1/4}e^{-\frac{p^{2}}{4\sigma_{pp}}}e^{\mp\frac{i\sigma_{xp}}{2\hbar
\sigma_{pp}}p^{2}}~, \label{sigmagauss2}%
\end{equation}
where $\sigma_{pp}>0$ is determined by the relation $\sigma_{xp}^{2}%
=\sigma_{xx}\sigma_{pp}-\frac{1}{4}\hbar^{2}$. Both functions $\psi_{+}$ and
$\psi_{-}=\psi_{+}^{\ast}$ and their Fourier transforms $\widehat{\psi}_{+}$
and $\widehat{\psi}_{-}$ satisfy the conditions $|\psi_{+}(x)|^{2}=|\psi
_{-}(x)|^{2}$ and $|\widehat{\psi}_{+}(p)|^{2}=|\widehat{\psi}_{-}(p)|^{2}$
showing that the Pauli problem does not have a unique solution. In fact the
covariance matrices determined by the states $|\psi_{\pm}\rangle$ are,
respectively,%
\[
\Sigma_{+}=%
\begin{pmatrix}
\sigma_{xx} & \sigma_{xp}\\
\sigma_{px} & \sigma_{pp}%
\end{pmatrix}
\text{ \ },\text{ }\Sigma_{-}=%
\begin{pmatrix}
\sigma_{xx} & -\sigma_{xp}\\
-\sigma_{px} & \sigma_{pp}%
\end{pmatrix}
\]
with $\sigma_{xp}=\sigma_{px}$ and $\sigma_{xx}\sigma_{pp}-\sigma_{xp}%
^{2}=\frac{1}{4}\hbar^{2}$; this yields two covariance ellipses $\Omega_{+}$
and $\Omega_{-}$with area $\pi\hbar$ defined by%
\begin{equation}
\dfrac{\sigma_{pp}}{2D}x^{2}\mp\frac{\sigma_{xp}}{D}px+\dfrac{\sigma_{xx}}%
{2D}p^{2}\leq1~,
\end{equation}
which are symmetric by the reflections $x\rightarrow-x$ or $p\rightarrow-p$.
The projections of these ellipsoids on the $x$ and $p$ axes are in both cases
the polar dual line segments $\Omega_{X}=[-\sqrt{2\sigma_{xx}},\sqrt
{2\sigma_{xx}}]$\ and $\Omega_{P}=[-\sqrt{2\sigma_{pp}},\sqrt{2\sigma_{pp}%
}\dot{]}$.

To deal with the multidimensional case it will be convenient to use some
material from the Wigner formalism.

\subsection{The Wigner and Fourier transforms of Gaussians}

We recall some well-known facts about Gaussian states and their Wigner
transform. For details, proofs and generalizations see for instance
\cite{Wigner} or \cite{Birk,Littlejohn,sisumu}. The most general Gaussian
wavefunction on $\mathbb{R}_{x}^{n}$ can be written%

\begin{equation}
\phi_{WY}(x)=\left(  \tfrac{1}{\pi\hbar}\right)  ^{n/4}(\det W)^{1/4}%
e^{-\tfrac{1}{2\hbar}(W+iY)x^{2}} \label{fay}%
\end{equation}
where $W$ and $Y$ are real symmetric $n\times n$ matrices with $W$ positive
definite. In the case $n=1$ and taking $W=\hbar/2\sigma_{xx}$ and $Y=0$ one
obtains the minimum uncertainty Gaussian
\[
\psi_{0}(x)=(2\pi\sigma_{xx})^{-1/4}e^{-|x|^{2}/4\sigma_{xx}}.
\]

The Wigner transform (\ref{wigtra}) of $\phi_{WY}$ is given by
\cite{Birk,Wigner,Littlejohn,Tronci}
\begin{equation}
W\phi_{WY}(z)=(\pi\hbar)^{-n}e^{-\tfrac{1}{\hbar}Gz\cdot z} \label{phagauss}%
\end{equation}
where $G$ is the symplectic symmetric positive definite matrix
\begin{equation}
G=%
\begin{pmatrix}
W+YW^{-1}Y & YW^{-1}\\
W^{-1}Y & W^{-1}%
\end{pmatrix}
~. \label{gsym}%
\end{equation}
That $G$ indeed is symplectic follows from the observation that $G=S^{T}S$
where
\begin{equation}
S=%
\begin{pmatrix}
W^{1/2} & 0\\
W^{-1/2}Y & W^{-1/2}%
\end{pmatrix}
\label{bi}%
\end{equation}
obviously is in $\operatorname*{Sp}(n)$.

Using standard formulas for the calculation of Gaussian integrals
(\textit{e.g.} Lemma 241 in \cite{Birkbis}) the Fourier transform of
$\phi_{WY}$ is given by
\begin{equation}
\widehat{\phi}_{WY}(p)=\left(  \tfrac{1}{\pi\hbar}\right)  ^{n/4}(\det
W)^{1/4}\det(W+iY)^{-1/2}e^{-\tfrac{1}{2\hbar}(W+iY)^{-1}p^{2}} \label{ft1}%
\end{equation}
where $\det(W+iY)^{-1/2}=\lambda_{1}^{-1/2}\cdot\cdot\cdot\lambda_{n}^{-1/2}$
the $\lambda_{j}^{-1/2}$ being the square roots with positive real parts of
the eigenvalues $\lambda_{j}^{-1}$ of $(W+iY)^{-1}$. Using the elementary
identity \cite{tzon}
\[
(W+iY)^{-1}=(W+YW^{-1}Y)^{-1}-iW^{-1}Y(W+YW^{-1}Y)^{-1}%
\]
which is easily checked multiplying on the right by $W+iY$, we see that in
fact%
\begin{equation}
\widehat{\phi}_{WY}(p)=e^{i\gamma}\phi_{W^{\prime}Y^{\prime}}(p)\text{
\ \textit{with} }\left\{
\begin{array}
[c]{c}%
W^{\prime}=(W+YW^{-1}Y)^{-1}\\
Y^{\prime}=-W^{-1}Y(W+YW^{-1}Y)^{-1}%
\end{array}
\right.  \label{fayf}%
\end{equation}
where $e^{i\gamma}$ ($\gamma$ real) is a constant phase factor.

Setting $\Sigma^{-1}=\tfrac{2}{\hbar}G$ where $G$ is the symplectic matrix
(\ref{gsym}) we can rewrite its Wigner transform (\ref{wigrho}) as
\begin{equation}
W\phi_{WY}(z)=(2\pi)^{-n}\sqrt{\det\Sigma^{-1}}e^{-\frac{1}{2}\Sigma
^{-1}z\cdot z}~. \label{ouafi}%
\end{equation}
The inverse of $G$ being readily calculated using the formula for the inverse
of a symplectic matrix \cite{Birk,Littlejohn} we get the explicit expression
\begin{equation}
\Sigma=\frac{\hbar}{2}%
\begin{pmatrix}
W^{-1} & -W^{-1}Y\\
-YW^{-1} & W+YW^{-1}Y
\end{pmatrix}
~. \label{sigma}%
\end{equation}
Writing $\Sigma$ in block-matrix form (\ref{defcovma}) yields the system of
matrix equations%
\begin{equation}
\Sigma_{XX}=\frac{\hbar}{2}W^{-1}\text{ },\text{ }\Sigma_{XP}=-\frac{\hbar}%
{2}W^{-1}Y\text{ },\text{\ }\Sigma_{PP}=\frac{\hbar}{2}(W+YW^{-1}Y)~.
\label{ident}%
\end{equation}
Note that this system is overcomplete. In fact, the knowledge of partial
covariance matrices allows one to determine the corresponding Gaussian state
by solving the two first equalities (\ref{ident}) in $W$ and $Y$ one gets
\begin{equation}
W=\frac{\hbar}{2}\Sigma_{XX}^{-1}\text{ \ },\text{\ \ }Y=-\Sigma_{XP}%
\Sigma_{XX}^{-1}~; \label{AY}%
\end{equation}
insertion in the third yields
\begin{equation}
\Sigma_{XP}^{2}=\Sigma_{PP}\Sigma_{XX}-\frac{\hbar^{2}}{4}I_{n\times n}~
\label{sigpx}%
\end{equation}
which is the matrix version of the RSUP. Notice that for given $\Sigma_{XX}$
and $\Sigma_{PP}$ the solution $\Sigma_{XP}$ is not unique. We will see
(Theorem \ref{Thm3}) that this non-uniqueness is related to the existence of
\textquotedblleft Pauli partners" in the reconstruction problem.

\subsection{The multidimensional case\label{secrec3}}

\subsubsection{Saturation of the RSUP}

To generalize these constructions to the multidimensional case we begin by
briefly discussing the saturation properties of the RSUP. Assume that
$\widehat{\rho}$ is a Gaussian quantum state, that is, a state with Wigner
distribution
\begin{equation}
W_{\widehat{\rho}}(z)=\left(  \tfrac{1}{2\pi}\right)  ^{n}(\det\Sigma
)^{-1/2}e^{-\frac{1}{2}\Sigma^{-1}z\cdot z} \label{wigrho}%
\end{equation}
where $\Sigma$ satisfies the quantum condition (\ref{Quantum}). The purity of
this state is
\begin{equation}
\mu(\widehat{\rho})=\operatorname*{Tr}(\widehat{\rho}^{2})=\left(  \frac
{\hbar}{2}\right)  ^{n}(\det\Sigma)^{-1/2}~. \label{purity}%
\end{equation}
In view of Williamson's symplectic diagonalization theorem \cite{Birk,sisumu}
there exists $S\in\operatorname*{Sp}(n)$ such that%
\begin{equation}
\Sigma=S^{T}DS\text{ },\text{ }D=%
\begin{pmatrix}
\Lambda & 0\\
0 & \Lambda
\end{pmatrix}
\text{\ \textit{and}\ }\Lambda=\operatorname*{diag}(\nu_{1},...,\nu_{n})
\label{will}%
\end{equation}
with the $\nu_{j}>0$ being the symplectic eigenvalues of $\Sigma$
(\textit{i.e}. the numbers $\pm i\nu_{j}$ are the eigenvalues of $J\Sigma$,
that is, those of the antisymmetric matrix $\Sigma^{1/2}J\Sigma^{1/2}$). The
quantum condition $\Sigma+\frac{i\hbar}{2}J\geq0$ is equivalent to $\nu
_{j}\geq\frac{1}{2}\hbar$ for $j=1,...,n$. The Robertson--Schr\"{o}dinger
inequalities are saturated, that is,
\begin{equation}
\sigma_{x_{j}x_{j}}\sigma_{p_{j}p_{j}}=\sigma_{x_{j},p_{j}}^{2}+\tfrac{1}%
{4}\hbar^{2} \label{RSsat}%
\end{equation}
for $1\leq j\leq n$, if and only if $\nu_{j}=\frac{1}{2}\hbar$ for all $j$,
and this can only be achieved by pure Gaussian states (see \cite{Adesso} and
\cite{Fu}). Formula (\ref{purity}), implying that $\widehat{\rho}$ is a pure
state if and only if $\det\Sigma=(\hbar/2)^{2n}$, means, taking the
factorization (\ref{will}) into account, that we must have $\nu_{1}^{2}%
\cdot\cdot\cdot\nu_{n}^{2}=(\hbar/2)^{2n}$; since $\nu_{j}\geq\frac{1}{2}%
\hbar$ for all $j$ we must in fact have $\nu_{1}=\cdot\cdot\cdot=\nu_{n}%
=\frac{1}{2}\hbar$ so that the covariance matrix has the very particular form
\begin{equation}
\Sigma=\frac{1}{2}\hbar S^{T}S\text{ \ },\text{ \ }S\in\operatorname*{Sp}(n)
\label{sigmasym}%
\end{equation}
(this is equivalent to saying that the covariance ellipsoid is a quantum
blob). The saturating states are thus those with Wigner distribution
\[
W_{\widehat{\rho}}(z)=(\pi\hbar)^{-n}e^{-\frac{1}{\hbar}(S^{T}S)^{-1}z\cdot z}%
\]
hence the state is the Gaussian $\phi_{WY}(x)$ defined by (\ref{fay}). Let
\[
\phi_{0}(x)=(\pi\hbar)^{-n/4}e^{-|x|^{2}/2\hbar}%
\]
be the standard (normalized) Gaussian state; its Wigner distribution is
\cite{Birk,Tronci,Littlejohn}
\begin{equation}
W\phi_{0}(z)=(\pi\hbar)^{-n}e^{-\frac{1}{\hbar}|z|^{2}} \label{wfi0}%
\end{equation}
hence $W_{\widehat{\rho}}(z)=W\phi_{0}(S^{-1}z)$ and it follows from the
symplectic covariance properties of the Wigner transform
\cite{Birk,Littlejohn} that the state is the Gaussian $\psi=\widehat{S}%
\phi_{0}$ where $\widehat{S}$ is a unitary operator (a metaplectic operator)
associated with $S\in\operatorname*{Sp}(n)$ via the metaplectic representation
of the symplectic group (see \cite{blob} or \cite{Birk,Littlejohn} for
detailed descriptions of this method; note that in particular this shows that
all pure Gaussian states can be obtained from each other using only the
metaplectic group, in fact a subgroup thereof \cite{ATMP}). This discussion
can be summarized as follows:%
\begin{gather}
\text{\textit{The saturation of the RSUP is equivalent to the statement
\textquotedblleft}}\Omega\text{ \textit{is a }}\label{blob}\\
\text{\textit{quantum blob\textquotedblright, i.e. there exists }}%
S\in\operatorname*{Sp}(n)\text{ \textit{such that }}\Omega=S(\mathcal{B}%
^{2n}(\sqrt{\hbar}))~.\nonumber
\end{gather}

Note that if $S=I_{2n\times2n}$ then $\Omega=\mathcal{B}^{2n}(\sqrt{\hbar})$
so that the corresponding Gaussian is a minimum uncertainty state saturating
the Heisenberg inequality. In fact, property (\ref{blob}) says that every
Gaussian can be reduced to such a minimal state using a symplectic
transformation \cite{blob}.

\subsubsection{The reconstruction theorem: the saturated case $P=X^{\hbar}$}

The following intertwining lemma will allow us to reduce the study of the
reconstruction problem to a \textquotedblleft canonical\textquotedblright%
\ form. Recall from Section \ref{secgen} that the matrices
\[
M_{L}=%
\begin{pmatrix}
L^{-1} & 0\\
0 & L^{T}%
\end{pmatrix}
\text{ \ },\text{ \ }\det L\neq0
\]
are symplectic .

\begin{lemma}
\label{LemmaML}Let $\Omega$ be the phase space ellipsoid defined by
$Mz^{2}\leq\hbar$, $M>0$. Let $\Pi_{X}$ and $\Pi_{P}$ be the orthogonal
projections\ of $\mathbb{R}_{z}^{2n}$ onto $\mathbb{R}_{x}^{n}$ and
$\mathbb{R}_{p}^{n}$. We have
\begin{equation}
(\Pi_{X}\times\Pi_{P})M_{L}=M_{L}(\Pi_{X}\times\Pi_{P}) \label{inter}%
\end{equation}
that is
\begin{equation}
\Pi_{X}(M_{L}(\Omega))=L^{-1}\Pi_{X}\Omega\text{ \ and \ }\Pi_{P}(M_{L}%
(\Omega))=L^{T}\Pi_{P}\Omega~. \label{rescaling}%
\end{equation}

\end{lemma}

\begin{proof}
The ellipsoid $M_{L}(\Omega)$ is defined by $M^{\prime}z^{2}\leq\hbar$ where
$M^{\prime}=(M_{L}^{T})^{-1}MM_{L}^{-1}$; a direct calculation shows that the
Schur complements $M^{\prime}/M_{PP}^{\prime}$ and $M^{\prime}/M_{XX}^{\prime
}$ are given by $M^{\prime}/M_{PP}^{\prime}=L^{T}(M/M_{PP})L$ and $M^{\prime
}/M_{XX}^{\prime}=L^{-1}(M/M_{XX})(L^{T})^{-1}$. Formula (\ref{rescaling})
follows using (\ref{boundb}) and (\ref{bounda}).
\end{proof}

Before we proceed to prove the main result of this section, let us recall
(\cite{Birk}, \S 2.1) that a block matrix
\[
M=%
\begin{pmatrix}
M_{XX} & M_{XP}\\
M_{PX} & M_{PP}%
\end{pmatrix}
\]
is symplectic if and only if its blocks satisfy the relations%
\begin{subequations}
\begin{gather}
M_{XX}^{T}M_{PP}-M_{PX}^{T}M_{XP}=I_{n\times n}\label{ABCD1}\\
M_{XX}^{T}M_{PX}\text{ \textit{and} }M_{XP}^{T}M_{PP}\text{ \textit{symmetric}%
.} \label{ABCD2}%
\end{gather}
Also recall (\ref{Dual}) that if $X=\{x:Ax^{2}\leq\hbar\}$ and $P=\{p:Bp^{2}%
\leq\hbar\}$ with $A,B$ symmetric and positive definite, then $(X,P)$ is a
saturated dual pair if and only if $AB=I_{n\times n}$.
\end{subequations}
\begin{theorem}
\label{Thm2}Let $X=\{x:Ax^{2}\leq\hbar\}$ and $X^{\hbar}=\{p:A^{-1}p^{2}%
\leq\hbar\}$ its quantum polar dual.

(i) The product $X\times X^{\hbar}$ contains exactly one quantum blob
$\Omega=S(\mathcal{B}^{2n}(\sqrt{\hbar}))$, $S\in\operatorname*{Sp}(n)$, with
orthogonal projections $X$ and $X^{\hbar}$ on $\mathbb{R}_{x}^{n}$ and
$\mathbb{R}_{p}^{n}$; that quantum blob is the ellipsoid with the largest
volume inscribed in the convex set $X\times X^{\hbar}$;

(ii) $\Omega$ is the covariance ellipsoid of the pure Gaussian state
\begin{equation}
\psi(x)=\left(  \tfrac{1}{2\pi}\right)  ^{n/4}(\det\Sigma_{XX})^{-1/4}%
e^{-\tfrac{1}{4}\Sigma_{XX}^{-1}x\cdot x} \label{psix}%
\end{equation}
where $\Sigma_{XX}=\frac{\hbar}{2}A^{-1}$.
\end{theorem}

\begin{proof}
The symplectic transformation $M_{A^{-1/2}}$ takes the dual pair $(X,X^{\hbar
})$ to the dual pair $(\mathcal{B}_{X}^{n}(\sqrt{\hbar}),\mathcal{B}_{P}%
^{n}(\sqrt{\hbar}))$:
\begin{equation}
(X^{\prime},X^{\prime\hbar})=M_{A^{-1/2}}(X\times X^{\hbar})=(\mathcal{B}%
_{X}^{n}(\sqrt{\hbar})\times\mathcal{B}_{P}^{n}(\sqrt{\hbar}))~. \label{mlxp1}%
\end{equation}
In view of Lemma \ref{LemmaML} above, this has the effect of replacing the
projections $X$ and $X^{\hbar}$ with $\mathcal{B}_{X}^{n}(\sqrt{\hbar})$ and
$\mathcal{B}_{P}^{n}(\sqrt{\hbar})$. By a simple symmetry argument it is seen
that the John ellipsoid (which is the inscribed ellipsoid with largest volume
\cite{Ball,Schneider}) of $\mathcal{B}_{X}^{n}(\sqrt{\hbar})\times
\mathcal{B}_{P}^{n}(\sqrt{\hbar})$ is the phase space ball $\mathcal{B}%
^{2n}(\sqrt{\hbar})$. In view of the uniqueness of the John ellipsoid there is
no other quantum blob contained in $X^{\prime}\times X^{\prime\hbar}$: assume
we can find $S^{\prime}\in\operatorname*{Sp}(n)$ such that $S^{\prime
(}\mathcal{B}^{2n}(\sqrt{\hbar}))\subset X^{\prime}\times X^{\prime\hbar}$.
Since $S^{\prime}$ is volume preserving $S^{\prime}(\mathcal{B}^{2n}%
(\sqrt{\hbar})$ has same volume as $\mathcal{B}^{2n}(\sqrt{\hbar})$ so we must
have $S^{\prime}(\mathcal{B}^{2n}(\sqrt{\hbar})=\mathcal{B}^{2n}(\sqrt{\hbar
})$. The orthogonal projections of $\mathcal{B}^{2n}(\sqrt{\hbar})$ on
$\mathbb{R}_{x}^{n}$ and $\mathbb{R}_{p}^{n}$ being $\mathcal{B}_{X}^{n}%
(\sqrt{\hbar})$ and $\mathcal{B}_{P}^{n}(\sqrt{\hbar})$, respectively, we
conclude that the covariance ellipsoid we are looking for is precisely
$\Omega=\mathcal{B}^{2n}(\sqrt{\hbar})$. It corresponds to the standard
Gaussian $\phi_{0}(x)=(\pi\hbar)^{-n/4}e^{-|x|^{2}/2\hbar}$ whose Wigner
distribution is given by
\begin{equation}
W\phi_{0}(z)=(\pi\hbar)^{-n}e^{-\frac{1}{\hbar}|z|^{2}}~.
\end{equation}
Returning to the original dual pair $(X,X^{\hbar})$ using (\ref{mlxp1}) the
covariance ellipsoid is here
\[
\Omega=M_{A^{1/2}}(\mathcal{B}^{2n}(\sqrt{\hbar}))=\{z:M_{A^{-1}}z\cdot
z\leq\hbar\}~.
\]
Specializing the transformation table (\ref{Table1}) to $S=M_{A^{-1}}$\ we
have the correspondences%
\begin{equation}%
\begin{tabular}
[c]{|l|l|l|l|}\hline
$\Omega$ & $\Sigma$ & $W_{\widehat{\rho}}$ & $\widehat{\rho}$\\\hline
$M_{A^{-1}}(\Omega)$ & $M_{A^{1/2}}\Sigma M_{A^{1/2}}$ & $W_{\widehat{\rho}%
}\circ M_{A}$ & $\widehat{M}_{A^{-1},0}\widehat{\rho}\widehat{M}_{A,0}%
$\\\hline
\end{tabular}
\end{equation}
hence the state with covariance matrix $\Omega$ is the squeezed Gaussian
$\psi$ with Wigner transform
\[
W\psi(z)=(\pi\hbar)^{-n}e^{-\frac{1}{\hbar}M_{A^{T}}M_{A}z\cdot z}=(\pi
\hbar)^{-n}e^{-\frac{1}{\hbar}M_{A^{2}}z\cdot z}~.
\]
Setting $G=M_{A^{T}}M_{A}=M_{A^{2}}$ we have
\[%
\begin{pmatrix}
W+YW^{-1}Y & YW^{-1}\\
W^{-1}Y & W^{-1}%
\end{pmatrix}
=%
\begin{pmatrix}
(A^{2})^{-1} & 0\\
0 & A^{2}%
\end{pmatrix}
\]
hence $W=(A^{2})^{-1}$ and $Y=0$. In view of formulas (\ref{ouafi}) and
(\ref{sigma}) the state we are looking for is
\[
\phi_{(A^{2})^{-1}0}(x)=\left(  \tfrac{1}{\pi\hbar}\right)  ^{n/4}(\det
A)^{-1/2}e^{-\tfrac{1}{2\hbar}(A^{2})^{-1}x\cdot x}~;
\]
taking formula (\ref{ident}) into account this can be rewritten
\[
\psi(x)=\phi_{(A^{2})^{-1}0}(x)=\left(  \tfrac{1}{2\pi}\right)  ^{n/4}%
(\det\Sigma_{XX})^{-1/4}e^{-\tfrac{1}{4}\Sigma_{XX}^{-1}x\cdot x}%
\]
where $\Sigma_{XX}=\frac{\hbar}{2}W^{-1}=\frac{\hbar}{2}A^{2}$.
\end{proof}

\subsubsection{The reconstruction theorem in the general case}

We now consider the case $X^{\hbar}\subset P$, $X^{\hbar}\neq P$.

\begin{theorem}
\label{Thm3}Let $X=\{x:Ax^{2}\leq\hbar\}$ and $P=\{p:Bp^{2}\leq\hbar\}$ be two
ellipsoids such that $X^{\hbar}\subset P$, $X\neq P$.

(i) The product $X\times P$ contains two quantum blobs, i.e. two (centered)
ellipsoids $\Omega_{+}$ and $\Omega_{-}$ such that $\Omega_{\pm}=S_{\pm
}(\mathcal{B}^{2n}(\sqrt{\hbar}))$ for some $S_{\pm}\in\operatorname*{Sp}(n)$
and whose orthogonal projections are $X$ and $P$. These ellipsoids are the
covariance ellipsoids of two pure Gaussian quantum states explicitly given by
the formula%
\[
\psi_{\pm}(x)=\left(  \tfrac{1}{2\pi}\right)  ^{n/4}(\det\Sigma_{XX}%
)^{-1/4}\exp\left[  -\left(  \frac{1}{4}\Sigma_{XX}^{-1}\pm\frac{i}{2\hbar
}\Sigma_{XP}\Sigma_{XX}^{-1}\right)  x^{2}\right]
\]
where $\Sigma_{XX}$ and $\Sigma_{XP}$ are the $n\times n$ matrices defined
by:
\[
\Sigma_{XX}=\frac{\hbar}{2}A^{-1}\text{ },\text{ }\Sigma_{XP}=\frac{\hbar}%
{2}(B^{-1}A^{-1}-I_{n\times n})^{1/2}~.
\]

(ii) Let $\Omega=\{z:\tfrac{1}{2}\Sigma^{-1}z\cdot z\leq1\}$ be the ellipsoid
with largest volume contained in $X\times P$ and having projections $X$ and
$P$; the quantum state with Wigner distribution
\[
W_{\widehat{\rho}}(z)=\left(  \tfrac{1}{2\pi}\right)  ^{n}(\det\Sigma
)^{-1/2}e^{-\frac{1}{2}\Sigma^{-1}z\cdot z}%
\]
is a mixed state with purity $\mu(\widehat{\rho})=\lambda_{j_{1}}^{2}%
\cdot\cdot\cdot\lambda_{j_{m}}^{2}$ where the $\lambda_{j_{k}}$ are the
eigenvalues of $AB$ that are smaller than one.
\end{theorem}

\begin{proof}
\textit{(i)} Let us determine the quantum blobs $\Omega=S(\mathcal{B}%
^{2n}(\sqrt{\hbar}))$ ($S\in\operatorname*{Sp}(n)$) contained in $X\times P$
and orthogonally projecting onto $X$ and $P$. These will determine the
functions $\psi$ we are looking for by the same procedure as in Theorem
\ref{Thm2} via their covariance matrix $\Sigma$. Setting $M=\frac{\hbar}%
{2}\Sigma^{-1}$ the condition $\Omega=S(\mathcal{B}^{2n}(\sqrt{\hbar}))$ is
equivalent to $M\in\operatorname*{Sp}(n)$, $M>0$. The symplecticity of $M$
allows us to easily invert $\Sigma$ and one finds, using (\ref{ABCD1}) and
(\ref{ABCD2}),
\[
\Sigma=%
\begin{pmatrix}
\Sigma_{XX} & \Sigma_{XP}\\
\Sigma_{PX} & \Sigma_{PP}%
\end{pmatrix}
=\frac{\hbar}{2}%
\begin{pmatrix}
M_{PP} & -M_{PX}\\
-M_{XP} & M_{XX}%
\end{pmatrix}
~.
\]
The orthogonal projection $\Omega_{X}$ is given by the inequality
$(M/M_{PP})x^{2}\leq\hbar$ (Lemma \ref{LemmaProj}), that is, taking again the
equalities (\ref{ABCD1}) and (\ref{ABCD2}) into account and using the fact
that $M_{XX}$, $M_{PP}>0$ and $M_{XP}^{T}=M_{PX}$,
\[
M/M_{PP}=(M_{XX}M_{PP}-M_{XP}M_{PP}^{-1}M_{PX}M_{PP})M_{PP}^{-1}=M_{PP}%
^{-1}~.
\]
By a similar argument we get $M/M_{XX}=M_{XX}^{-1}$ hence the equalities
\begin{equation}
A=M/M_{PP}=\frac{\hbar}{2}\Sigma_{XX}^{-1}\text{ \ \textit{and} \ }%
B=M/M_{XX}=\frac{\hbar}{2}\Sigma_{PP}^{-1}~. \label{ABM}%
\end{equation}
It follows that the orthogonal projections $\Omega_{X}$ and $\Omega_{P}$ are
the ellipsoids
\[
\Omega_{X}=\{x:\tfrac{1}{2}\Sigma_{XX}^{-1}x^{2}\leq1\}\text{ \ },\text{
\ }\Omega_{P}=\{p:\tfrac{1}{2}\Sigma_{PP}^{-1}p^{2}\leq1\}~.
\]
We next determine all the Gaussian states $\phi_{WY}$ having $\Omega$ as
covariance matrix. As in the proof of \ref{Thm2}, we have to solve the matrix
equation%
\[%
\begin{pmatrix}
W+YW^{-1}Y & YW^{-1}\\
W^{-1}Y & W^{-1}%
\end{pmatrix}
=%
\begin{pmatrix}
M_{XX} & M_{XP}\\
M_{PX} & M_{PP}%
\end{pmatrix}
~.
\]
The solutions are (\textit{cf.} formulas (\ref{AY})) $W=\frac{\hbar}{2}%
\Sigma_{XX}^{-1}$ and $Y=-\Sigma_{XP}\Sigma_{XX}^{-1}$ corresponding to the
Gaussian pure state
\begin{multline*}
\phi_{WY}(x)=\left(  \tfrac{1}{2\pi}\right)  ^{n/4}(\det\Sigma_{XX})^{-1/4}\\
\exp\left[  -\left(  \frac{1}{4}\Sigma_{XX}^{-1}+\frac{i}{2\hbar}\Sigma
_{XP}\Sigma_{XX}^{-1}\right)  x\cdot x\right]
\end{multline*}
where $\Sigma_{XP}$ is any matrix satisfying condition the matrix version
(\ref{sigpx}) of the RSUP, that is%
\[
\Sigma_{XP}^{2}=\Sigma_{PP}\Sigma_{XX}-\frac{\hbar^{2}}{4}I_{n\times n}~;
\]
Since $\Sigma_{XX}=\frac{\hbar}{2}A^{-1}$ and $\Sigma_{PP}=\frac{\hbar}%
{2}B^{-1}$ (formulas (\ref{ABM}) above) this is%
\[
\Sigma_{XP}^{2}=\frac{\hbar^{2}}{4}(B^{-1}A^{-1}-I_{n\times n})
\]
and we are done. \textit{(ii)}\ We can, as in the proof of Theorem
\ref{Theorem1}, choose an invertible $n\times n$ matrix $L$ such that
\[
L^{T}AL=L^{-1}B(L^{T})^{-1}=\Lambda~.
\]
In view of Lemma \ref{LemmaML} above, replacing $(X,P)$ with
\begin{equation}
X^{\prime}\times P^{\prime}=M_{L}(X\times P)\text{ \ , \ }M_{L}=%
\begin{pmatrix}
L^{-1} & 0\\
0 & L^{T}%
\end{pmatrix}
\label{mlxp}%
\end{equation}
has the effect of replacing the projections $\Omega_{X}$ and $\Omega_{P}$ of
an ellipsoid $\Omega$ with $L^{-1}\Omega_{X}$ and $L^{T}\Omega_{P}$. This
reduces the proof to the case where $X$ and $P$ are replaced with
\begin{equation}
X^{\prime}=\Lambda^{-1/4}\mathcal{B}_{X}^{n}(\sqrt{\hbar})\text{
\ },\text{\ \ }P^{\prime}=\Lambda^{-1/4}\mathcal{B}_{P}^{n}(\sqrt{\hbar})
\label{reduc}%
\end{equation}
where $\Lambda=\operatorname*{diag}(\sqrt{\lambda_{1}},...,\sqrt{\lambda_{n}%
})$, the $\lambda_{j}$ being the eigenvalues of $AB\leq I_{n\times n}$; the
duality of $X$ and $P$ (and hence of $X^{\prime}$ and $P^{\prime}$) is
equivalent to the conditions $0<\lambda_{j}\leq1$ for $j=1,...,n$ with at
least one of the eigenvalues $\lambda_{j}$ of $AB$ being $<1$ since $X^{\hbar
}\neq P$ implies that $AB\leq I_{n\times n}$, $AB\neq I_{n\times n}$.
Explicitly:
\[
X^{\prime}=\{x:%
{\textstyle\sum_{j=1}^{n}}
\lambda_{j}^{1/2}x_{j}^{2}\leq\hbar\}\text{ },\text{ }P^{\prime}=\{p:%
{\textstyle\sum_{j=1}^{n}}
\lambda_{j}^{1/2}p_{j}^{2}\leq\hbar\}~.
\]
Now, the John ellipsoid of $X^{\prime}\times P^{\prime}$ is
\[
\Omega_{\max}^{\prime}=\{(x,p):%
{\textstyle\sum_{j=1}^{n}}
\lambda_{j}^{1/2}(x_{j}^{2}+p_{j}^{2})\leq\hbar\}
\]
and the associated covariance matrix is
\[
\Sigma_{\max}^{\prime}=\frac{\hbar}{2}%
\begin{pmatrix}
\Lambda^{-1/2} & 0\\
0 & \Lambda^{-1/2}%
\end{pmatrix}
~.
\]
The purity (\ref{purity}) of the associated Gaussian state $\widehat{\rho}$
is
\[
\mu(\widehat{\rho})=\left(  \frac{\hbar}{2}\right)  ^{n}(\det\Sigma_{\max
}^{\prime})^{-1/2}=\det\Lambda^{2}%
\]
and this is the square of the product of the eigenvalues of $AB$ that are
smaller than one.
\end{proof}

\section{The Mahler Volume and Related Topics}

In this section we briefly discuss some related topics where quantum polar
duality can also be seen to appear, sometimes unexpectedly. Particularly
interesting is the link between quantum mechanics and a well-known conjecture
from convex geometry, the \textit{Mahler conjecture}. Perhaps the
Donoho--Stark uncertainty principle which is discussed thereafter might shed
some new light on this difficult problem.

\subsection{The Mahler conjecture\label{secmahl}}

\subsubsection{Some known results}

Let $X$ be a convex body in $\mathbb{R}_{x}^{n}$ (i.e. $X$ is compact and has
non-empty interior). We assume that $X$ contains $0$ in its interior. By
definition, the Mahler volume \cite{Mahler} of $X$ is the product%
\begin{equation}
\upsilon(X)=|X|~|X^{\hbar}| \label{Mahler}%
\end{equation}
where $|X|$ is the usual Euclidean volume on $\mathbb{R}_{x}^{n}$. The Mahler
volume is a dimensionless quantity because of its rescaling invariance 8see
below): we have $\upsilon(\lambda X)=\upsilon(X)$ for all $\lambda>0$.

The Mahler volume is invariant under linear automorphisms of $\mathbb{R}%
_{x}^{n}$: if $L$ is an automorphism of $\mathbb{R}_{x}^{n}$ then we have, in
view of the scaling formula (\ref{scaling}),%
\begin{equation}
\upsilon(LX)=|LX|~|(L^{T})^{-1}X^{\hbar}|=|X|~|X^{\hbar}|~. \label{mahlerinv}%
\end{equation}
It follows that the Mahler volume of an arbitrary ellipsoid $X=\{x:Ax^{2}\leq
R^{2}\}$ ($A>0$) is given by
\begin{equation}
\upsilon(X)=|\mathcal{B}^{n}(\sqrt{\hbar})|~|\mathcal{B}^{n}(\sqrt{\hbar
})^{\hbar}|=\frac{(\pi\hbar)^{n}}{\Gamma(\frac{n}{2}+1)^{2}} \label{volxvolxh}%
\end{equation}
and is thus the same for all ellipsoids. It turns out that the Mahler volume
of ellipsoids is maximal, in the sense that we have
\begin{equation}
\upsilon(X)\leq\frac{(\pi\hbar)^{n}}{\Gamma(\frac{n}{2}+1)^{2}} \label{BS}%
\end{equation}
for all symmetric convex bodies, with equality occurring if and only if $X$ is
an ellipsoid. This result is due to Blaschke \cite{Blaschke} for $n=2,3$ and
to Santal\'{o} \cite{Santalo} for arbitrary $n$ (see Schneider
\cite{Schneider}).

The problem of finding a lower bound for the Mahler volume is much more
difficult and a general solution is unknown. A famous conjecture, due to
Mahler himself \cite{Mahler}, says that for every symmetric convex body $X$ in
$\mathbb{R}_{x}^{n}$ we have
\begin{equation}
\upsilon(X)\geq\frac{(4\hbar)^{n}}{n!} \label{volvo3}%
\end{equation}
with equality only when $X$ is the hypercube $C=[-1,1]^{n}$. In view of the
invariance property (\ref{mahlerinv}) this is tantamount to saying that the
minimum is attained by any $n$-parallelepiped%
\begin{equation}
X=[-\sqrt{2\sigma_{x_{1}x_{1}}},\sqrt{2\sigma_{x_{1}x_{1}}}]\times\cdot
\cdot\cdot\times\lbrack-\sqrt{2\sigma_{x_{n}x_{n}}},\sqrt{2\sigma_{x_{n}x_{n}%
}}] \label{interval}%
\end{equation}
which is the $n$-dimensional generalization of the interval $\Omega_{X}$
(\ref{intervals}) of the introduction. While the conjectured inequality
(\ref{volvo3})\ trivially holds when $n=1$ (since $\upsilon(X)$ is just the
area of the rectangle $X\times X^{\hbar}$), a proof in the general case is
still lacking at the time of writing. Bourgain and Milman \cite{BM} have shown
the existence, for every $n\in\mathbb{N}$, of a constant $C_{n}>0$ such that
\begin{equation}
|X|~|X^{\hbar}|\geq C_{n}\hbar^{n}/n! \label{BM}%
\end{equation}
and more recently Kuperberg \cite{Kuper} has shown that one can choose
$C_{n}=(\pi/4)^{n}$, so that (\ref{BM}) can be rewritten%
\begin{equation}
\upsilon(X)\geq\frac{(\pi\hbar)^{n}}{4^{n}n!} \label{kuper}%
\end{equation}
and this is the best known lower bound for the Mahler volume. Summarizing, we
have the bounds%
\begin{equation}
\frac{(\pi\hbar)^{n}}{4^{n}n!}\leq\upsilon(X)\leq\frac{(\pi\hbar)^{n}}%
{\Gamma(\frac{n}{2}+1)^{2}}~. \label{bounds}%
\end{equation}

One geometric meaning of the Mahler volume is that it captures the
\textquotedblleft roundness\textquotedblright\ of a convex body, with
ellipsoids being the roundest, and cubes and octahedra being the
\textquotedblleft pointiest\textquotedblright\ \cite{Tao}. It is clear that
this lower bound -- ideally, the conjectured bound $\upsilon(X)\geq
(4\hbar)^{n}/n!$ -- is a form of the uncertainty principle. But what does it
tell us?

\subsubsection{Mahler volume and symplectic capacity}

We know that the notion of symplectic capacity is closely related to the
uncertainty principle. There is an important inequality relating the
symplectic capacity of a symmetric convex body $K$ to its volume. It is the
so-called \textit{symplectic isoperimetric inequality} \cite{aa,arkaos13}
which says that%
\begin{equation}
\frac{c_{\min}(K)}{c_{\min}(\mathcal{B}^{2n}(1))}\leq\left(  \frac
{|K|}{|\mathcal{B}^{2n}(1)|}\right)  ^{1/n} \label{iso1}%
\end{equation}
where $c_{\min}$ is the Gromov width; in other words%
\begin{equation}
c_{\min}(K)\leq(n!)^{1/n}|K|^{1/n}~. \label{iso2}%
\end{equation}
The proof of (\ref{iso1})--(\ref{iso2}) is quite simple: let $\mathcal{B}%
^{2n}(r)$ be the largest phase space ball that can be embedded in $K$ using a
canonical transformation, thus $c_{\min}(\Omega)=\pi r^{2}$. Since canonical
transformations are volume preserving we have also $|K|\geq|\mathcal{B}%
^{2n}(r)|$ hence the inequality $\mathcal{B}^{2n}(r)$ follows by a direct
calculation. Since all symplectic capacities agree on ellipsoids the
inequality (\ref{iso1}) still holds when $K$ is an ellipsoid and $c_{\min}$ is
replaced with any symplectic capacity $c$. It is conjectured
(\textquotedblleft Viterbo's conjecture\textquotedblright) that (\ref{iso1})
actually holds for \emph{all} convex bodies and \emph{all} symplectic
capacities:%
\begin{equation}
c(K)\leq(n!)^{1/n}|K|^{1/n} \label{conjecture}%
\end{equation}
(see \cite{arkaos13} for details and references). Quite surprisingly, this
inequality implies the Mahler conjecture. In fact, if (\ref{conjecture})
holds, then we may choose $c=c_{\max}$ and hence, by formula (\ref{yaron}) in
Theorem \ref{Theorem1},%
\[
4\hbar=c_{\max}(X\times X^{\hbar})\leq(n!)^{1/n}|X\times X^{\hbar}|^{1/n}%
\]
that is $\upsilon(X)\geq(4\hbar)^{n}/n!$, which is the inequality
(\ref{volvo3}) conjectured by Mahler.

\subsection{Hardy's Uncertainty Principle}

Let $\psi\in L^{2}(\mathbb{R})$, $||\psi||_{L^{2}}\neq0$. Hardy's uncertainty
principle \cite{Hardy} in its original form states that we cannot have
simultaneously
\begin{equation}
|\psi(x)|\leq Ce^{-ax^{2}/2\hbar}\text{ \ },\text{ \ }|\widehat{\psi}(p)|\leq
Ce^{-bp^{2}/2\hbar} \label{H1}%
\end{equation}
($a,b,C$ positive constants) unless $ab\leq1$ and \textit{(i)} if $ab=1$ then
$\psi(x)=\alpha e^{-ax^{2}/2\hbar}$ for some $\alpha\in\mathbb{C}$ and
\textit{(ii)} if $ab<1$ then $\psi$ is a finite linear combination of
conveniently rescaled Hermite functions.

In the multidimensional case Hardy's uncertainty principle can be stated as
follows \cite{golu09}: Let $A$ and $B$ be positive definite and symmetric
matrices and $\psi\in L^{2}(\mathbb{R}^{n})$, $||\psi||_{L^{2}}\neq0$. The
Hardy inequalities
\begin{equation}
|\psi(x)|\leq Ce^{-\tfrac{1}{2\hbar}Ax^{2}}\text{ \ and \ }|\widehat{\psi
}(p)|\leq Ce^{-\tfrac{1}{2\hbar}Bp^{2}} \label{AB}%
\end{equation}
are satisfied for some constant $C>0$ if and only if $AB\leq I_{n\times n}$,
that is,
\begin{equation}
\text{\textit{The eigenvalues} }\lambda_{1},...,\lambda_{n}\text{ \textit{of}
}AB\text{ \textit{are}}\leq1 \label{eigen1}%
\end{equation}
and we have:

(i) \textit{If} $\lambda_{j}=1$ \textit{for all} $j$\textit{, then}
$\psi(x)=\alpha e^{-\frac{1}{2\hbar}Ax^{2}}$ \textit{for some constant
}$\alpha\in\mathbb{C}$;

(ii) \textit{If} $\lambda_{j}<1$ \textit{for at least one index} $j$\textit{,}
\textit{then the set of functions satisfying} (\ref{AB})\textit{ is an
infinite-dimensional subspace of }$L^{2}(\mathbb{R}^{n})$\textit{.}

In view of property (\ref{Dual}) the conditions (\ref{eigen1}) mean that the
ellipsoids
\[
X_{A}=\{x:Ax^{2}\leq\hbar\}\text{ \textit{and} \ }P_{B}=\{p:Bp^{2}\leq\hbar\}
\]
form a dual quantum pair $(X_{A},P_{B})$. If this pair is saturated
(\textit{i.e.} $P_{B}=X_{A}^{\hbar}$), then $\psi$ is a scalar multiple of the
Gaussian $\phi_{AY}$. Consider now the \textquotedblleft Hardy
ellipsoid\textquotedblright\
\[
\Omega_{AB}=\{(x,p):Ax^{2}+Bp^{2}\leq\hbar\}
\]
that is
\[
\Omega_{AB}=\{z:M_{AB}z^{2}\leq\hbar\}\text{ \ , \ }M_{AB}=%
\begin{pmatrix}
A & 0\\
0 & B
\end{pmatrix}
~.
\]
The orthogonal projections on $\mathbb{R}_{x}^{n}$ and $\mathbb{R}_{p}^{n}$ of
$\Omega_{AB}$ are precisely the ellipsoids $X_{A}$ and $P_{B}$. The symplectic
eigenvalues of $M_{AB}$ are the positive numbers $\nu_{1},...,\nu_{n}$ such
that $\pm i\nu_{1},...,\pm\nu_{n}$ are the solutions of the characteristic
polynomial $P(t)=\det(t^{2}I_{n\times n}+AB)$ of $M$. These are the pure
imaginary numbers $\pm i\sqrt{\lambda_{1}},...,\pm i\sqrt{\lambda_{n}}$ where
the $\lambda_{j}>0$ are the eigenvalues of $AB$. Thus $\nu_{j}=\sqrt
{\lambda_{j}}$ for $1\leq j\leq n$. Since we have $\lambda_{j}\leq1$ for all
$j$ the covariance matrix $\Sigma_{AB}=\frac{\hbar}{2}M_{AB}^{-1}$ satisfies
the quantum condition\ $\Sigma_{AB}+\frac{i\hbar}{2}J\geq0$; equivalently
(\ref{foop}): $c(\Omega_{AB})\geq\pi\hbar$. If, in particular, the
$\lambda_{j}$ are all equal to one we have $c(\Omega_{AB})=\pi\hbar$ and
$AB=I_{n\times n}$ so that $P_{B}=X_{A}^{\hbar}$. Let us examine this case a
little bit closer at the light of the reconstruction Theorems above. Assume
that $\psi\in L^{2}(\mathbb{R}^{n})$, $||\psi||_{L^{2}}\neq0$, and its Fourier
transform satisfy
\begin{equation}
|\psi(x)|\leq Ce^{-\tfrac{1}{2\hbar}Ax^{2}}\ \text{\textit{and} }%
\ |\widehat{\psi}(p)|\leq Ce^{-\tfrac{1}{2\hbar}A^{-1}p^{2}} \label{esta}%
\end{equation}
for some constant $C>0$. The ellipsoids $X_{A}$ and $P_{A^{-1}}$ are polar
dual of each other hence Theorem \ref{Thm2} tells us that
\[
\psi(x)=\left(  \tfrac{1}{\pi\hbar}\right)  ^{n/4}(\det A)^{1/4}e^{-\tfrac
{1}{2\hbar}Ax\cdot x}~.
\]
The Fourier transform of $\psi$ is given (up to a constant factor with modulus
one) by
\[
\widehat{\psi}(p)=\left(  \tfrac{1}{\pi\hbar}\right)  ^{n/4}(\det A^{-}%
)^{1/4}e^{-\tfrac{1}{2\hbar}A^{-1}x\cdot x}\phi_{W^{\prime}Y^{\prime}}(p)
\]
\ (formula (\ref{fayf})) and the inequalities (\ref{esta}) are satisfied since
we have%
\[
A(A+YA^{-1}Y)^{-1}\leq I_{n\times n}~.
\]
A similar argument allows to to study the general case $AB\leq I_{n\times n}$
using Theorem \ref{Thm2}. Hardy's uncertainty principle thus appears as being
a particular case of the reconstruction theorems we have proven, and which are
themselves based on the notion of quantum polar duality.

\subsection{Donoho and Stark's uncertainty principle}

As we mentioned in the introduction, Hilgevoord and Uffink emphasized in
\cite{hi02,hiuf85bis} that standard deviations only give adequate measurements
of the spread for Gaussian states. A good candidate for a more general theory
of indeterminacy is to define an uncertainty principle using the notion of
concentration of a state. It turns out that Donoho and Stark \cite{dost} have
proven a concentration result for a function and its Fourier transform which
can be viewed in a sense as a variant of Hardy's uncertainty principle; as we
will see it can also be interpreted in terms of quantum polar duality and is
related to the Mahler volume. Let $X\subset\mathbb{R}_{x}^{n}$ be a measurable
set and let $\overline{X}=\mathbb{R}_{x}^{n}\setminus X$ be its complement
(convexity is not assumed here). We will say that a function $\psi\in
L^{2}(\mathbb{R}^{n})$ is $\varepsilon$-concentrated on $X$ if we have
\begin{equation}
\left(  \int_{\overline{X}}|\psi(x)|^{2}dx\right)  ^{1/2}\leq\varepsilon
||\psi||_{L^{2}}~. \label{dos1}%
\end{equation}
If $||\psi||_{L^{2}}=1$, which we assume from now on, this is equivalent to
the inequality
\begin{equation}
\int_{\overline{X}}|\psi(x)|^{2}dx\leq\varepsilon^{2}~. \label{dos2}%
\end{equation}
The Donoho--Stark uncertainty principle says that if the normalized function
$\psi\in L^{2}(\mathbb{R}_{x}^{n})$ is $\varepsilon_{X}$-concentrated on $X$
and its Fourier transform $\widehat{\psi}$ is $\varepsilon_{P}$-concentrated
of $P$, that is
\begin{equation}
\int_{\overline{X}}|\psi(x)|^{2}dx\leq\varepsilon_{X}^{2}\text{ \ },\text{
}\int_{\overline{P}}|\widehat{\psi}(p)|^{2}dp\leq\varepsilon_{P}^{2}
\label{ksiksi}%
\end{equation}
then we must have
\begin{equation}
|X|~|P|\geq(2\pi\hbar)^{n}(1-\varepsilon_{X}-\varepsilon_{P})^{2} \label{DS}%
\end{equation}
for $\varepsilon_{X}+\varepsilon_{P}<1$. Taking $P=X^{\hbar}$ this shows in
particular that the Mahler volume of $X$ satisfies
\[
\upsilon(X)\geq(2\pi\hbar)^{n}(1-\varepsilon_{X}-\varepsilon_{X^{\hbar}}%
)^{2}~.
\]

Let us apply the estimate above to the dual pair $(X,X^{\hbar})$ of centrally
symmetric convex bodies. We have the following remarkable result relating the
Donoho--Stark UP and the Mahler volume:

\begin{theorem}
Let $X$ be a symmetric convex measurable body in $\mathbb{R}_{x}^{n}$ and
$\psi\in L^{2}(\mathbb{R}_{x}^{n})$, $||\psi||_{L^{2}}=1$. Assume that $\psi$
is $\varepsilon_{X}$-concentrated in $X$ and $\widehat{\psi}$ is
$\varepsilon_{X^{\hbar}}$-concentated in $X^{\hbar}$ with $\varepsilon
_{X}+\varepsilon_{X^{\hbar}}\leq1$. Then we must have%
\begin{equation}
1\geq\varepsilon_{X}+\varepsilon_{X^{\hbar}}\geq1-\frac{1}{2^{n/2}\Gamma
(\frac{n}{2}+1)} \label{good}%
\end{equation}
that is $\varepsilon_{X}+\varepsilon_{X^{\hbar}}\rightarrow1$ as
$n\rightarrow\infty$.
\end{theorem}

\begin{proof}
Combining the Blaschke--Santal\'{o} estimate (\ref{BS}) for the Mahler volume
$\upsilon(X)=|X|~|X^{\hbar}|$ and the Donoho--Stark inequality (\ref{DS}) we
get%
\[
\frac{(\pi\hbar)^{n}}{\Gamma(\frac{n}{2}+1)^{2}}\geq(2\pi\hbar)^{n}%
(1-\varepsilon_{X}-\varepsilon_{X^{\hbar}})^{2}%
\]
and hence%
\[
0\leq1-\varepsilon_{X}-\varepsilon_{X^{\hbar}}\leq\frac{1}{2^{n/2}\Gamma
(\frac{n}{2}+1)}%
\]
which is (\ref{good}).
\end{proof}

If the Mahler volume of $X$ satisfies the equality
\[
\upsilon(X)=(2\pi\hbar)^{n}(1-\varepsilon_{X}-\varepsilon_{X^{\hbar}})^{2}%
\]
then we must have
\begin{equation}
1-\frac{1}{2^{n/2}\Gamma(\frac{n}{2}+1)}\leq\varepsilon_{X}+\varepsilon
_{X^{\hbar}}\leq1-\frac{1}{8^{n/2}n!^{1/2}}~. \label{doublegood}%
\end{equation}
This follows from the estimate (\ref{bounds}) for the Mahler volume.

These estimates show that when the number of degrees of freedom $n$ is large,
the sum $\varepsilon_{X}+\varepsilon_{X^{\hbar}}$ of the concentrations of a
wavefunction and of its Fourier transform is practically equal to one. If the
Mahler conjecture is true, then (\ref{doublegood}) may be replaced with
\begin{equation}
1-\frac{1}{2^{n/2}\Gamma(\frac{n}{2}+1)}\leq\varepsilon_{X}+\varepsilon
_{X^{\hbar}}\leq1-\frac{2}{(2\pi)^{n/2}(n!)^{1/2}}~.
\end{equation}
For example, if $n=6$ (which corresponds to a system of two particles moving
in physical space) we will have $0.979<\varepsilon_{X}+\varepsilon_{X^{\hbar}%
}<0.999$.

\begin{acknowledgement}
This work has been financed by the Grant P 33447 of the Austrian Research
Agency FWF. It is my pleasure to extend my thanks to Basil Hiley and Glen
Dennis for useful comments and for having pointed out typos. I also express my
gratitude\ to the Reviewer for very useful remarks.
\end{acknowledgement}

\end{document}